\documentclass[1pt]{elsarticle}
\usepackage{geometry}
\usepackage{graphicx,subfigure}
\usepackage[colorlinks,linkcolor=red,citecolor=blue]{hyperref}
\usepackage{amsmath,amssymb,amsxtra,bm}
\usepackage{mathrsfs,amsfonts,amsthm}
\usepackage{appendix}
\usepackage{epstopdf}
\usepackage{booktabs,multirow}
\usepackage{float}
\usepackage{dsfont}
\usepackage{stmaryrd}
\usepackage{xcolor}
\usepackage{cases}
\usepackage[ruled]{algorithm2e}

\newtheorem{remark}{Remark}
\newtheorem{example}{Example}

\journal{??}

\newtheorem{sch}{Scheme}[section]
\theoremstyle{plain}
\newtheorem{thm}{Theorem}[section]

\theoremstyle{remark}

\def\d{\mathrm{d}}
\def\p{\partial}
\def\btau{\bm{ \tau }}
\def\bsig{\bm{ \sigma }}
\def\ol{\overline}

\def\fr{\frac{1}{2}}

\def\half{n+\frac{1}{2}}
\def\d{\mathrm{d}}

\def\bF{\mathbf{F}}

\def\bp{\mathbf{p}}

\def\bn{\mathbf{n}}
\def\bx{\bm{x}}

\def\bv{\mathbf{v}}

\def\bu{\mathbf{u}}

\def\wt{\widetilde}
\def\wh{\widehat}
\def\p{\partial}
\def\b0{\bm{0}}

\def\cA{\mathcal{A}}
\def\cB{\mathcal{B}}
\def\bg{\mathbf{g}}

\def\bD{\mathbf{D}}

\def\bI{\mathbf{I}}

\def\bR{\mathbf{R}}

\def\bV{\mathbf{V}}

\def\bbR{\mathbb{R}}
\def\bR{\mathbf{R}}
\def\dist{\mathrm{dist}}
\def\hf{n + \frac{1}{2}}

\usepackage{cases}
\def\drag{\mathrm{drag}}
\def\ex{\mathrm{ex}}

\newcommand{\ben}{\begin{eqnarray}}
\newcommand{\een}{\end{eqnarray}}
\newcommand{\beq}{\begin{equation}}
\newcommand{\eeq}{\end{equation}}
\newcommand{\bea}{\begin{array}}
\newcommand{\eea}{\end{array}}
\newcommand{\bef}{\begin{figure}[H]}
\newcommand{\eef}{\end{figure}}

\numberwithin{equation}{section}

\numberwithin{equation}{section}
\def\balc#1\ealc{\begin{align}\begin{cases}#1\end{cases}\end{align}}
\def\bseq#1\eseq{\begin{subequations}\begin{numcases}{}#1\end{numcases}\end{subequations}}

\begin{document}

\begin{frontmatter}

\title{A hybrid phase field method for fluid-structure interactions in viscous fluids}

\author[mymainaddress,mymainaddress2,mymainaddress3]{Qi Hong}
\author[mysecondaddress]{Qi Wang}
\cortext[mycorrespondingauthor]{Corresponding author}
\ead{qwang@math.sc.edu}
\address[mymainaddress]{College of Science, Nanjing University of Aeronautics and Astronautics, Nanjing 210016, China}
\address[mymainaddress2]{Key Laboratory of Mathematical Modelling and High Performance Computing of Air Vehicles (NUAA), MIIT, Nanjing  211106, China}
\address[mymainaddress3]{{Jiangsu Key Laboratory for Numerical Simulation of Large Scale Complex System}}
\address[mysecondaddress]{Department of Mathematics, University of South Carolina, Columbia, SC 29208, USA}
\begin{abstract}

We present a novel computational modeling framework to numerically investigate fluid-structure interaction in viscous fluids using the phase field embedding method. Each rigid body or elastic structure  immersed in the incompressible viscous fluid matrix, grossly referred to as the particle in this paper,  is identified by a volume preserving phase field. The motion of the particle is driven by the  fluid velocity in the matrix for passive particles or combined with  its self-propelling velocity for active particles. The excluded volume effect between a pair of particles or between a particle and the boundary is modeled by a repulsive potential force. The drag exerted to the fluid by a particle is assumed proportional to its velocity. When the particle is rigid, its state is described by a zero velocity gradient tensor within the nonzero phase field that defines its profile and a constraining stress exists therein. While the particle is elastic, a linear constitutive equation for the elastic stress is provided within the particle domain.  A hybrid, thermodynamically consistent hydrodynamic model valid in the entire computational domain is then derived for the fluid-particle ensemble  using the generalized Onsager principle accounting for both rigid and elastic particles. Structure-preserving numerical algorithms are subsequently developed for the thermodynamically consistent model. Numerical tests in 2D and 3D space are carried out to verify the rate of convergence and numerical examples are given to demonstrate the usefulness of the computational framework for simulating fluid-structure interactions for passive as well as self-propelling active particles in a viscous fluid matrix.
\end{abstract}

\begin{keyword}
Phase field, hydrodynamics, hybrid model, active matter, fluid-structure interaction, structure-preserving algorithm.
\end{keyword}

\end{frontmatter}

\section{Introduction}
\label{sec:Introduction}
Fluid-structure interaction (FSI)  phenomena  exist ubiquitously in nature and arise in many scientific and engineering settings.
Well-known examples include air flows passing air-craft wings,  blood flows through arteries carrying red blood cells, the response of bridges and tall buildings to winds, the vibration of turbine and compressor blades while operating, the oscillation of heat exchangers etc. For interested readers, please refer to \cite{hou2012numerical} for a detailed review. In order to better understand the underlying physical mechanism in these phenomena arising from FSI, one needs to find suitable ways to model and simulate the multiphase materials systems involved in the phenomena.

Modeling of FSI problems is  challenging mathematically and computationally.
Traditionally,  one studies FSI problems using two distinctive material's descriptions for the grossly two phase material system: the Eulerian description for the fluid and the Lagrange description for the solid structure.
The Arbitrary Lagrangian-Eulerian (ALE) method \cite{Hughes1981lagrangian,Sarrate2001Arbitrary,Dettmer2006comp} is a popular approach to dealing with FSI problems, where the computational domain is subdivided into a fluid domain that is described in an Eulerian coordinate system  and a solid structure domain that is usually treated in a Lagrange description.
Both meshes are aligned at the fluid-solid interface which typically moves with the calculated velocity. However, the computational domain have to re-meshed or updated as the solid moves or deforms, which induces a quite expensive computational cost. Another popular class of methods to study FSI is the interface-tracking methods,
in which the fluid, solid and the interface are described by a labeling or phase variable. The most popular interface-tracking methods include the level-set \cite{Legay2006Euler,Cottet2008Euler}, volume-of-fluid \cite{Hirt1981volume,Sugiyama2011full} and front-tracking method \cite{Trygvason2001JCP}. But the implementation of the interface tracking method in three-dimensional space can be very complicated, computational expensive and sometimes computationally prohibitive.

Recently, the phase-field approach \cite{Gomez2008Iso}, i.e., the diffuse interface method, has become a popular choice for modeling multi-phase  materials in materials and life sciences (cf. for instance \cite{Gong_SISC2018,Kou2018Thermo,Lin2006Sim,Liu2017Dynamic,Wang-JCP2011,Xu2019three,Yu2017Numerical}). Although the phase-field models are widely used in studying multi-phase materials systems (e.g., alloys, complex fluids, etc.), their use in fluid-structure interaction problem has not been explored extensively \cite{Sun2014full,mokbel2018a}. Soft bodies in viscous fluid flows such as red blood cells in blood flows, have been modeled and simulated using phase field models in the past \cite{Lin2021}. However, truly rigid bodies in viscous fluid matrix has not been explored rigorously in the context of nonequilibrium thermodynamics.

When a hydrodynamic phase field model is derived by a variational method based on the  Onsager principle for multiphase fluid flows \cite{OnsagerL1, OnsagerL2, Wang2021Book}, the governing hydrodynamic phase field system naturally admits an energy dissipation law and is therefore thermodynamically consistent and mathematically well-posed.  The variational structure and property of the modeling approach has been exploited to develop structure-preserving, energy stable numerical solvers, where the numerical scheme preserving some structural  properties of the model is called the geometric integrator or structure-preserving numerical algorithm.

When the energy dissipation property in the model is preserved at the discrete level while  the equations of the model are discretized.  A numerical scheme that preserves this property is known as an energy stable scheme. There have been quite a few efficient numerical methods developed to obtain energy stable algorithms for solving thermodynamical and hydrodynamical phase-field models, for instances, the fully-implicit structure-preserving schemes \cite{Boyer2011Numer,Xu2019On}, the convex splitting schemes \cite{Guan&Lowengrub&Wang&WiseJCP2014, Shen&Wang&Wise2012, WangDCDS2010, WiseSINUMA2009}, stabilizer technique \cite{Li2016Characterizing, SY_DCDSA, Gao2018Decoupled, Yang09}. In the past a few years, the energy quadratization  method and the subsequent   scalar auxiliary variable method have
fueled the development for energy stable schemes \cite{ ZhaoJ2017,ShenJ2018,Cheng2020new, Lin2020gPAV,Yang2020roadmap,Yang2021CMAME}. However, when some of these methods are applied to solving thermodynamically consistent models,  the resulting schemes warrant  a modified energy dissipation law instead of  the original energy dissipative law. Recently, \cite{Hong2020SVM,Gong2020SVM} proposed a general supplementary variable method outlining a strategy to design structure-preserving schemes for any PDE systems with deduced equations to ensure that the deduced equations are consistent with the discretized equations. It offers  a great deal of flexibilities in the development of structure-preserving numerical algorithms.

In this paper, we present a new modeling framework for developing mathematical models describing fluid structure interactions between rigid particle structures in a viscous fluid matrix. The rigid ``structures" are represented by phase fields which can be either fixed, driven by the flow or self-propelled while immersed in the viscous fluid matrix. The structure can be of an arbitrary morphology in principle. While immersed in the fluid matrix, the structure exerts a drag to the fluid and there exists a repulsive force between each pair of the structures when they move close to each other. The density of the structures can be different and different from the surrounding fluid matrix. For simplicity and illustration purposes, we assume the viscous fluid is described by a Brinkman model and densities of the structures and the fluid are identical. We defer the discussion of the quasi-incompressible model for systems with distinct densities to a sequel because the numerical treatment would be quite different and more challenging for those models. In the current study, two types of efficient and energy-production-rate preserving numerical approximations, including the coupled and decoupled schemes, are proposed to solve the newly developed hybrid, hydrodynamical phase-field model, where the linear-implicit Crank-Nicolson scheme and the staggered-grid finite difference method is utilized to discretize the governing system of equations in time and space, respectively. We remark that additional energy-dissipation rate preserving schemes based on BDF2 and higher order collocation methods can be devised as well. In the end, we further investigate the stability and accuracy of the schemes numerically and active particle interactions for a number of numerical examples in 2D and 3D space.

The paper is organized as follows. In \S2, we present a systematic and detailed derivation of the hybrid, thermodynamically consistent hydrodynamical phase field model for fluid-particle interaction employing the generalized Onsager principle. In \S3, we develop two second-order thermodynamically consistent numerical approximations to this model.  Then, we prove that the schemes preserve the energy dissipation rate. In \S4,  a series of numerical experiments in 2D and 3D space are presented to show the accuracy, efficiency  and usefulness of the newly developed schemes in FSI simulations.
Finally, some concluding remarks are given in \S5.

\section{Model Formulation}

Suppose that $\Omega \in \bbR^d$ with $d = 2, 3$ is a bounded, open and connected domain, whose boundary is piecewise smooth and denoted by $\p \Omega$. We denote the $L^2$ inner product of any two functions $f(\bx)$ and $g(\bx)$ in $L_2(\Omega)$ by  $(f(\bx), g(\bx)) = \int_{\Omega}f(\bx) g(\bx) \d \bx$ and the $L^2$ norm of the function $f(\bx)$ by $\|f\| = \sqrt{ (f(\bx), f(\bx)) }$.

We present two models for rigid and elastic particles immersed in an incompressible viscous fluid matrix. The first is for passive rigid particles and the second for active rigid particles. We then remark that this framework applies to particles made of soft matter as well.

\subsection{Model for passive particles  in an incompressible viscous fluid  matrix}

We consider  a set of passive rigid or elastic  particles immersed in an incompressible  viscous fluid matrix governed by the extended Brinkman equation. We remark that the choice of Brinkman model is for convenience in this paper because we are interested in applications of situations where the fluid matrix is  more viscous. A more general model for viscous fluid governed by the Navier-Stokes model can be derived as well.  The particles modeled can be in any shapes, represented by labeling functions or phase field variables $\phi_i,\; i=1, \cdots, N,$ where $N$ is the number of the particles. We want to represent the mass conservation and the momentum conservation of the fluid-particle ensemble using one set of equations, valid in both the fluid and the particle domain. For simplicity, we assume the density of the fluid matches that of the particle in this paper.
The governing system of equations for the fluid flow and the particle ensemble is given by
\ben
\bea{l}
\nabla \cdot  \bv=0,\\[0.3cm]
\rho \dfrac{\p \bv}{\p t}=\mu \nabla^2 \bv - \nabla p + \nabla \cdot \btau + \bF,
\eea
\een
where $\rho$ is the density of the fluid, $\bv$ is the mass average velocity, $\btau$ is the extra stress including the constraining stress to maintain the particles' rigidity or elasticity,  and $\bF$ is the body force acted on the fluid by the particles as well as external forces.

We model the $i$-th particle as a rigid body with a shape characterized by phase variable $\phi_i(x,t)$, $i=1,\cdots, N$.
The $i$-th phase field is given by
\ben
\phi_i=
\left \{
\bea{ll}
1,& \; \hbox{ in particle $i$,}\\
0, & \; \hbox{ outside the $i$-th particle through a diffuse interface of width $\epsilon$.}
\eea
\right.
\een
We denote  the domain of the fluid and particle ensemble by $\Omega$ and the center of mass of the $i$-th particle by
\ben
\bR_i(t)=\frac{1}{\int_{\Omega} \phi_i \d\bx} \int_{\Omega} \bx \phi_i \d\bx.
\een
We assume $\int_{\Omega} \phi_i \d\bx$ is a constant.  The motion of the center of mass for the $i$-th particle is governed by
\ben
\frac{\d}{\d t} \bR_i(t)=\bV_i(t), \quad i=1, \cdots, N,
\een
where
\ben
\bV_i=\frac{1}{\int_{\Omega} \phi_i \d\bx}\int_{\Omega} \phi_i \bv(x, t) \d\bx.
\een
In the rigid particle, we require the rate of strain   $\bD=\frac{1}{2}(\nabla \bv+\nabla \bv^T)$ to be  zero
\ben
\phi_i \bD = \b0, \quad i=1,\cdots, N. \label{D-eq}
\een
This allows the rigid body to rotate freely but not to deform. If we were to restrict the rigid body not to rotate freely, we have to impose
\ben
\phi_i \nabla \bv = \b0, \quad i=1,\cdots, N.\label{gradv-eq}
\een
\eqref{D-eq} and \eqref{gradv-eq} are additional constraints besides the mass and momentum conservation equation in the model.

The forces exerted by the particles on the fluid include the repulsive  force, the elastic force, the external body force and the drag force, respectively,	
\begin{align}\label{all-force-eq}
	\begin{cases}
		\bF = \sum_{i=1}^N \left[\bF_{i, \ex} + \bF_{i, b} + \bF_{i, \drag}\right]\frac{\phi_i}{\int_{\Omega} \phi_i \d \bx} + \bF_{e},\quad
		\bF_{i, \ex} = -\sum_{j=1}^N \nabla_{\bR_i} V(\bR_i, \bR_j),\\[0.3cm]
		\bF_{i, b} = -\sum_{j=1}^N \nabla_{\bR_i} V_b(\bR_i, \partial \Omega),\quad
		\bF_{i, \drag} = -\alpha_i \bV_i,\\[0.3cm]
		V =
		\begin{cases}	
			\epsilon_2 [(\frac{R_m}{R_{ij}})^{12}-2(\frac{R_m}{R_{ij}})^6], \quad & R_{ij}\leq R_m,\\[0.3cm]
			0, \quad &  R_{ij}>R_m,
		\end{cases}\\[0.3cm]
		V_b =
		\begin{cases}	
			\epsilon_2 [(\frac{R_n}{R_{ib}})^{12}-2(\frac{R_n}{R_{ib}})^6], \quad & R_{ib}\leq R_n,\\[0.3cm]
			0, \quad &  R_{ib}>R_n.
		\end{cases}
	\end{cases}
\end{align}	
Here $\bF_e$ includes both the interfacial force and the external body force, $\alpha_i$ is the drag coefficient on the $i$-th particle, $\epsilon_2$ parameterizes the strength of the excluded volume potential, $R_{i j}=\|\bR_i-\bR_j\|$ is the distance between $\bR_i$ and $\bR_j$, $\bR_{i b}= \dist(\bR_i, \partial \Omega)$, $R_m$ and $R_n$ are prescribed distance parameters. An external body force can be added if necessary.
We note that the drag force and the one for excluded volume only depend on time and concentrated at the center of mass of the particle in this model.

For incompressible fluids and the particle with a matching density, we assume the density  is constant. We will consider the case where the densities are different in subsequent studies.  The transport equation for the $i$-th phase field  variable is given by
\begin{align}\label{phi-eq}
	\frac{\p }{\p t}\phi_{i } + \nabla \cdot (\phi_{i } \bv) = j_i,
\end{align}
where production rate $j_i$ is to be determined below.  For each phase field, we assign a free energy to it denoted as $F_i$ with density $  f_i(\phi_i, \nabla \phi_{ i } , \nabla \nabla \phi_i, \cdots)$.  We define the   total energy of the system as follows
\ben
E_{total}=\int_{\Omega} \left[ \frac{\rho}{2} \|\bv\|^2+ \sum_{i = 1}^n f_i(\phi_i, \nabla \phi_{ i }, \nabla \nabla  \phi_i, \cdots )  \right] \d \bx + \fr \sum_{i \neq j = 1}^{N} V(\bR_i, \bR_j)+\sum_{j=1}^N V_b(\bR_{j},\partial \Omega).
\een	
The time rate of change of the total energy is  calculated as follows
\begin{align}\label{energy-change-eq}
	\dfrac{\d E_{total}}{\d t}
	& = \int_{\Omega} \left[ \rho \bv \cdot \bv_t + \sum_{i = 1}^N  \left( \frac{\p f}{\p \phi_{i}} \phi_{i, t} + \frac{\p f}{\p \nabla \phi_{i }} \cdot \nabla \phi_{i, t }+  \frac{\p f}{\p \nabla \nabla  \phi_{i }}:\nabla \nabla  \phi_{i, t } +\cdots\right) \right]\d \bx \notag\\
	&\quad + (\sum_{i \neq j =1}^{N} \nabla_{\bR_i} V(\bR_i, \bR_j)+\sum_{i=1}^N \nabla_{\bR_i} V_{b}) \cdot \frac{\d \bR_i(t)}{\d t}\notag\\
	&  = \int_{\Omega} \left[ \bv \cdot \nabla \cdot (-p \bI  + b )  + \bv \cdot \bF
	+ \sum_{i = 1}^N  \left(\frac{\p f}{\p \phi_{i}} \phi_{i, t} + \frac{\p f}{\p \nabla \phi_{i }}\cdot  \nabla \phi_{i, t } +\frac{\p f}{\p \nabla \nabla \phi_{i }}:\nabla \nabla \phi_{i, t } +\cdots \right) \right]\d \bx \notag\\
	&\quad + (\sum_{i \neq j =1}^{N} \nabla_{\bR_i} V(\bR_i, \bR_j)+\sum_{i=1}^N \nabla_{\bR_i} V_b) \cdot \frac{\d \bR_i(t)}{\d t} \notag \\
	& = (p, \nabla \cdot \bv) - \int_{\partial \Omega }\bn \cdot (p \bv) \d s - (\btau, \nabla \bv ) + \int_{\p \Omega} \bn \cdot \btau \bv \d s + (\bF, \bv) +  \sum_{i = 1}^N \bigg( \frac{\p f}{\p \phi_{i } } - \nabla \cdot \left( \frac{\p f}{\p \nabla \phi_i}\right) \notag\\
	&\quad+\nabla \nabla: \left( \frac{\p f}{\p \nabla \nabla  \phi_i}\right)+\cdots,\; \phi_{i, t }  \bigg) +
	\sum_{i = 1}^N  \int_{\partial \Omega } \bn \cdot \bigg(  \frac{\p f}{\p \nabla \phi_i} \phi_{i, t}+\frac{\p f}{\p \Delta \phi_i}\cdot \nabla \phi_{i,t}-\nabla \cdot \frac{\p f}{\p \nabla \nabla  \phi_i}\phi_{i,t} \notag\\
	&\quad+\cdots \bigg) \d s + (\sum_{i \neq j =1}^{N} \nabla_{\bR_i} V(\bR_i, \bR_j)+\sum_{i=1}^N \nabla_{\bR_i}V_b) \cdot \frac{\d \bR_i(t)}{\d t} \notag \\
	& = - \int_{\partial \Omega }\bn \cdot (p \bv) \d s - (\btau, \nabla \bv )  + \int_{\p \Omega} \bn \cdot \btau \bv \d s + (\bF, \bv)
	+  \sum_{i = 1}^N \left( \mu_i, \phi_{i, t }  \right) \notag\\
	&\quad+(\sum_{i \neq j =1}^{N} \nabla_{\bR_i} V(\bR_i, \bR_j)+\nabla_{\bR_i}V_b) \cdot \frac{\d \bR_i(t)}{\d t} \notag \\
	&\quad + \sum_{i = 1}^N \int_{\partial \Omega } \bn \cdot \left(  \frac{\p f}{\p \nabla \phi_i} \phi_{i, t}+\frac{\p f}{\p \nabla \nabla  \phi_i}\cdot \nabla \phi_{i,t}-\nabla \cdot \frac{\p f}{\p \nabla \nabla  \phi_i}\phi_{i,t}+\cdots\right) \d s.
\end{align}
Here, $\bn$ is the unit external normal of the boundary and $\mu_i=\frac{\delta f_i}{\delta \phi_i}$ is the chemical potential for phase variable $\phi_i$.
From \eqref{all-force-eq}, one obtains	
\begin{align}\label{F-v-inner-pro-eq}
	(\bF, \bv) + \sum_{i \neq j =1}^{N} \nabla_{\bR_i} V(\bR_i, \bR_j)\cdot  \frac{\d \bR_i(t)}{\d t} = -\sum_{i=1}^N \alpha_i \bV_i^2  + (\bF_e, \bv).
\end{align}
We set $\btau = \sum_{i  }( 2 \eta_i \phi_{ i } \bD  +  \phi_{ i } \bsig_i )$ following the Onsager principle \cite{Wang2021Book}.  It follows from \eqref{phi-eq} that
\begin{align}\label{tau-inn-prod-eq}
	\begin{cases}
		-(\btau,\;  \nabla \bv) = -\sum_{i } (2 \eta_i \phi_{i } \bD, \bD) - \sum_{i }(\bsig_i, \phi_{i }\bD),\\
		(\mu_i, \phi_{i, t} ) = (\mu_i, j_i - \nabla \cdot (\phi_{i } \bv)) = (\mu_i, j_i) + (\phi_i \nabla \mu_i, \bv) - \int_{\p \Omega} \bn \cdot \left( \mu_i \phi_{i } \bv \right) \d s.
	\end{cases}
\end{align}
Substituting \eqref{F-v-inner-pro-eq} and \eqref{tau-inn-prod-eq} into \eqref{energy-change-eq} yields
\begin{align}\label{dE-dt-eq}
	\dfrac{\d E_{total}}{\d t}  =&  -\sum_{i=1}^N (2 \eta_i \phi_{i } \bD, \bD) - \sum_{i=1}^N(\bsig_i, \phi_{i }\bD) + \sum_{i=1}^N (\mu_i, j_i) + \left( \bF_e + \sum_{i=1}^N \phi_{ i } \nabla \mu_i,\; \bv \right)\\
	&-\sum_{i=1}^N \alpha_i  \bV_i ^2 + \int_{\p \Omega} g_1 \d s,
\end{align}
where $g_1$ is given by
\begin{align}
	g_1 = \bn \cdot \left[ -p \bv + \btau \bv + \sum_{i =1}^N \left( \frac{\p f}{\p \nabla \phi_{ i }}  \phi_{i, t }+\frac{\p f}{\p \nabla \nabla  \phi_i}\cdot \nabla \phi_{i, t}-\nabla \cdot \frac{\p f}{\p \nabla \nabla  \phi_i}\phi_{i, t}+\cdots  -  \mu_i \phi_{i } \bv \right) \right].
\end{align}
Next, we apply $\phi_i \bD = 0$ and the Onsager principle to obtain constitutive equations
\begin{align}
	\bF_e  = -\sum\limits_{i = 1}^N \phi_i \nabla \mu_i , \quad
	j_i = \nabla \cdot (M_i \nabla \mu_i),
\end{align}
where $M_i$ is the mobility operator and the external body force is assumed absent.

Then, \eqref{dE-dt-eq} reduces to
\begin{align}
	\dfrac{\d E_{total}}{\d t} = - \sum_{i=1}^{N} \alpha_i \| \bV_i \|^2  -\int_{\Omega}\sum_{i=1}^N [2 \eta_i \phi_i  \bD: \bD+\nabla \mu_i M_i \nabla \mu_i] \d \bx + \int_{\p \Omega} g \d s,
\end{align}
where the surface energy density $g$ is given  by
\begin{align*}
	g = \bn \cdot \left[ -p \bv + \btau \cdot \bv  + \sum_{i = 1}^N \left( \frac{\p f}{\p \nabla \phi_{i } } \phi_{i, t }+\frac{\p f}{\p \nabla \nabla \phi_i}\cdot \nabla \phi_{i,t}-\nabla \cdot \frac{\p f}{\p \nabla \nabla  \phi_i}\phi_{i,t}+\cdots  +  \mu_i M_i \nabla \mu_i - \bv \phi_{i } \mu_i \right) \right].
\end{align*}

If we adopt the adiabatic boundary conditions for $i=1, 2, \cdots, N$ as follows
\begin{align}\label{incomp-BCs}
	\bv=0|_{\p \Omega},\;\;
	\bn\cdot \frac{\partial f_i}{\partial \nabla \phi_i}\bigg|_{\p \Omega} = 0,\;\;
	\bn \cdot \left(\nabla \cdot \frac{\p f}{\p \nabla \nabla \phi_i}\right) = 0,\;\;
	\bn \cdot \frac{\p f}{\p \nabla \nabla \phi_i}  = 0,\;\;
	\bn\cdot M_i \nabla \mu_i|_{\p \Omega} = 0,
\end{align}
or  periodic boundary conditions for all variables,
the energy dissipation rate reduces to
\begin{align}\label{alternative-disp}
	\frac{\d E_{total}}{\d t} &= - \sum_{i=1}^{N} \alpha_i  \bV_i^2  -\int_{\Omega}\sum_{i=1}^N [2 \eta_i \phi_i  \bD : \bD+\nabla \mu_i M_i \nabla \mu_i] \d \bx\\
\end{align}
The energy dissipation rate is non-positive semi-definite provided the mobility operators are  non-negative semi-definite:
\ben
M_i \geq 0,\quad
i=1,\cdots, N.
\een

\begin{remark}
	With the Cahn-Hilliard transport equation of constant mobility $M_i$ for phase variable $\phi_i,\; i=1,\cdots, N$ and  no-flux boundary conditions, $\int_{\Omega} \phi_i \d\bx$ is a constant and
	\begin{align}
		\frac{\d \bR_i(t)}{\d t}&=\frac{1}{\int_{\Omega} \phi_i \d\bx} \int_{\Omega} \bx [-\nabla \cdot (\phi_i \bv)+\nabla \cdot M_i \nabla \mu_i] \d\bx\\
		&=-\frac{1}{\int_{\Omega} \phi_i \d\bx} \int_{\Omega} \nabla \bx \cdot  [-(\phi_i \bv)+ M_i \nabla \mu_i] \d\bx\\
		&=\frac{1}{\int_{\Omega} \phi_i \d\bx} \int_{\Omega}  [\phi_i \bv- M_i \nabla \mu_i] \d\bx\\[0.1cm]
		&=\bV_i(t)
	\end{align}
	with respect to periodic boundary conditions or $\int_{\partial \Omega} \bn M_i \mu_i \d \bx=0$. The later condition is satisfied should the particle is away from the boundary.
\end{remark}

\begin{remark}
	In many applications, one has to impose inflow and outflow boundary conditions, i.e., $\bv\neq0$ in a portion of the boundary $\partial \Omega$. In this case, the incompressibility condition implies
	\ben
	\int_{\partial \Omega} \bn\cdot \bv \d s=0.
	\een
	This indicates that the inflow and outflow rate with respect to the interested domain $\Omega$ must equal. In this case, so long as the particles are not at the inflow and outflow boundary, their phase volume is kept at a constant should the above boundary conditions for the phase fields are held.
\end{remark}

Finally, we summarize the governing system of equations  for the fluid-particle ensemble as follows
\begin{align}\label{incomp-model-eq}
	\begin{cases}
		\rho \bv_t = - \nabla p + \nabla \cdot \btau + \bF, \\
		\btau = \sum_{i=1}^N ( 2 \eta_i \phi_{ i } \bD   + \phi_{i } \bsig_i ), \\
		\nabla \cdot \bv = 0,\\
		\phi_{i, t } + \nabla \cdot(\phi_{i} \bv ) = \nabla \cdot (M_i \nabla \mu_i),\\
		\phi_{ i } \bD = \b0,\;  i=1,\cdots, N,
	\end{cases}
\end{align}
where $\mu_i = \delta F_i / \delta \phi_{ i }$ represents the chemical potential for the $i$-th phase-field  and the total body force $\bF$ is given by \eqref{all-force-eq}.

\subsection{Phase field with a nonlocal inertia perturbation}

In the following, we consider a perturbed transport equation of the phase field given by
\ben\label{pf-perturbed}
(1+sM_i \partial_t \Delta^2)\phi_{i, t } + \nabla \cdot(\phi_{i} \bv ) =j_i=\nabla \cdot (M_i \nabla \mu_i)
\een
where $s>0$ is a parameter and $M_i>0$ is a constant. We set
\ben
\bF_e=-\sum_i \phi_i \nabla \hat{\mu}_i,
\een
where $\hat{\mu}_i$ is defined by
\ben\label{hat-mu-i}
\hat{\mu}_i=\frac{\delta f_i}{\delta \phi_i}-s \Delta \p_t \phi_{i, t}.
\een
Note that
\begin{align*}
	(\mu_i, \phi_{ i, t })
	&= (\hat{\mu}_i, \phi_{i, t}) - s(\nabla \phi_{i, t}, \p_t \nabla \phi_{i, t})   \\
	& = \left(\hat{\mu}_i, \nabla\cdot(M_i \nabla \mu_i) - sM_i \p_t \Delta^2 \phi_{ i, t } - \nabla \cdot(\phi_{ i }\bv) \right) - s(\nabla \phi_{i, t}, \p_t \nabla \phi_{i, t}) \\
	& = (\hat{\mu}_i, \nabla \cdot (M_i \nabla \hat{\mu}_i)) - (\wh{\mu}_i, \nabla\cdot(\phi_i \bv)) - s(\nabla \phi_{i, t}, \p_t \nabla \phi_{i, t})  \\
	& = -(\nabla\hat{\mu}_i,   M_i \nabla \hat{\mu}_i) + (\phi_i \nabla \hat{\mu}_i, \bv) - s(\nabla \phi_{i, t}, \p_t \nabla \phi_{i, t}).
\end{align*}
Thus,  the energy dissipation rate for this modified model is given by
\begin{align}
	\dfrac{\d E_{total}}{\d t} & = -\sum_{i}(2 \eta_i \phi_{ i } \bD, \bD) + \sum_{i}(\mu_i, \phi_{ i, t }) - \sum_{i}\alpha_{i}\bV_i^2 + (\bF_{e}, \bv)\\
	&=
	-\sum_{i=1}^{N} \alpha_i  \bV_i^2  -\int_{\Omega}\sum_{i=1}^N [2 \eta_i \phi_i  \bD : \bD+\nabla \hat{\mu}_i M_i \nabla \hat{\mu}_i +
	s \nabla \phi_{i,t} \p_t \nabla \phi_{i,t} ] \d \bx.
\end{align}

So, if we define a modified energy
\ben
\wh{E}=E_{total} + \sum_{i=1}^N\frac{s}{2} \|\nabla \phi_{i,t}\|^2,
\een
\ben
\frac{\d\wh{E}}{\d t}=- \sum_{i=1}^{N} \alpha_i  \bV_i^2  -\int_{\Omega}\sum_{i=1}^N [2 \eta_i \phi_i  \bD : \bD+\nabla \hat{\mu}_i M_i \nabla \hat{\mu}_i] \d \bx.
\een
The governing system of equations is once again given by \eqref{incomp-model-eq} with $\mu_i$ replaced by $\hat{\mu}_i$ and the phase field equations replaced by \eqref{pf-perturbed}.
This perturbed system serves as the foundation to design the so-called stabilized numerical scheme.
\begin{remark}
	If the particle is not rigid, we would drop the rigidity constraint on the velocity gradient tensor or the rate of strain tensor and assign large viscosity to the particle regions. The model so-derived is  also thermodynamically consistent.
\end{remark}
It is worthy noting that  the constraint of the embedded $i$-th rigid body: $\phi_i {\bf D} = \b0$ does not have a time derivative. This equation determines the constraining stress $\bsig_i$. We next explore an extended system with an added time derivative for this equation that relaxes back to the equation in a limit.

\subsection{Elastic  relaxation }

We relax the rigid body constraint by adding a time derivative for $\bsig_i$ in the constraint for the rigid particles as follows
\begin{align}
	\varepsilon \bsig_{i, t} = \phi_{ i }\bD,
\end{align}
where $\varepsilon$ is a  small positive parameter whose inverse is the elastic modulus.  In the particle region,  we use the same viscosity as that of the fluid and the mobility for each phase is the same constant.
Then, the modified governing equation consists of the following  equations:
\begin{align}\label{relax-model-eq}
	\begin{cases}
		\frac{\d \bR_i}{\d t} = \bV_i(t),\quad  i = 1,\cdots, N,\\
		\rho \bv_t =-\nabla p+\eta \Delta \bv + \nabla \cdot\left( \sum_{i =1}^N\phi_{ i }\bsig_i \right) + \sum_{i=1}^N [\bF_{i, \ex}+\bF_{i, \drag}]\frac{\phi_i}{\int_{\Omega} \phi_i \d\bx} -\sum_{i=1}^N \phi_i \nabla \hat{\mu}_i ,\\
		\nabla \cdot \bv = 0,\\
		(1+sM_i \partial_t \Delta^2)\phi_{i, t} + \nabla \cdot (\phi_i \bv) = M \Delta \hat{\mu}_i,\\
		\varepsilon \bsig_{i, t} - \phi_i \bD = {\bm 0}.
	\end{cases}
\end{align}
We define the free energy for the extended model as follows
\ben\label{relax-energy-eq}
\bea{l}
\wh{E} =\int_{\Omega} \left[ \frac{\rho}{2} |\bv|^2+\sum_{i=1}^N f_i(\phi_i, \nabla \phi_i) + \sum_{i=1}^N \frac{\varepsilon}{2} |\bsig_i|^2 + \sum_{i=1}^{N}\frac{s}{2}|\nabla \phi_{ i, t }|^2 \right] \d \bx \\
~\quad+ \fr \sum_{i \neq j = 1}^{N} V(\bR_i, \bR_j)+\sum_{i=1}^N V_b(\bR_i, \partial \Omega).
\eea
\een

The following theorem assures that the  energy dissipation law of the modified hydrodynamics model  is valid for the   boundary conditions  given above and there is no boundary condition for $\bsig_i$ necessary here.
\begin{thm}
	Model \eqref{relax-model-eq} with  physical boundary conditions \eqref{incomp-BCs} possesses the following energy dissipation law:
	\begin{align}
		\dfrac{\d \wh{E} }{\d t} = - \sum_{i  = 1}^N \left[ (2 \eta_i \phi_{i } \bD, \bD) +  M \|\nabla \hat{\mu}_i \|^2  + \alpha_i \bV_i^2 \right],
	\end{align}
	where the total energy $\wh{E}$ is defined in \eqref{relax-energy-eq}.
\end{thm}

\begin{remark}
	Note that when $\varepsilon \rightarrow 0$, this relaxed hydrodynamical model in \eqref{relax-model-eq} reduces to the original model in \eqref{incomp-model-eq}. Therefore, this extended or relaxed  model is indeed a generalization of the original hydrodynamical phase field model for fluid-rigid particle interaction.  In addition, by means of adding time relaxation dynamics in $\bsig_i$, the boundary conditions of the new model \eqref{relax-model-eq} remains unchanged.
\end{remark}

\begin{remark}
	The free energy for the phase field model must be chosen to maintain rigidity of the particle while the velocity is held a constant within the particle. In addition, the viscosity within each particle region can also be made larger to facilitate its rigidity without influencing dynamics of the flow field outside the rigid particles.
\end{remark}

\subsection{Model for active particles   in an  incompressible viscous fluid  matrix}
When the rigid particle is  self-propelling, the system becomes an active matter system. We denote the self-propelling velocity of the $i$-th particle by $\bp_i(t)$. The kinematic equation for the center of mass of the $i$-th particle is modified into
\ben\label{active-speed}
\frac{\d \bR_i}{\d t}=\bV_i(t)+\bp_i,\quad  i=1,\cdots, N.
\een
The drag force has to be changed accordingly so that the force balance equation along with other relevant equations become
\begin{align}
	\begin{cases}
		\rho \bv_t =-\nabla p+\eta \Delta \bv + \nabla \cdot\left( \sum_{i=1  }^N\phi_{ i }\bsig_i \right) +  \sum_{i=1}^N [\bF_{i, \ex}+\bF_{i,b}+\bF_{i, \drag}]\frac{\phi_i}{\int_{\Omega} \phi_i \d\bx} -\sum_{i=1}^N \phi_i \nabla \hat{\mu}_i ,\\[0.2cm]
		\bF_{i, \drag  } = -\alpha_i (\bV_i+\bp_i),\\[0.2cm]
		\phi_{i, t} + \nabla \cdot (\phi_i (\bv+\bp_i)) = M \Delta \hat{\mu}_i.
	\end{cases}
\end{align}
In the active matter system, the total energy may no longer be dissipative. The energy production rate is given by
\ben
\frac{\d \wh{E}}{\d t}  = -\eta \| \nabla \bv \|^2 - \sum_{i  = 1}^N \bigg[  M \|\nabla \hat{\mu}_i \|^2  + \alpha_i \| \bV_i\|^2-(\sum_{j =1, j \neq i}^N \nabla_{\bR_i} V(\bR_i, \bR_j)+\sum_{i}\nabla_{\bR_i}V_b(\bR_i, \partial \Omega)) \bp_i\\
+\alpha_i \bV_i \cdot \bp_i + \hat{\mu}_i \nabla \cdot (\phi_i \bp_i)  \bigg].\notag
\een
In the following numerical implementation, we focus on circular particles in 2D and spherical particles in 3D space. We choose the free energy for the ith particle with a conformational part and a double well bulk part given by
$f_i(\phi_{i }, \nabla \phi_{i}) = \frac{\gamma_1 }{ 2} |\nabla \phi_{i }|^2 + \gamma_2 \phi_i^2(1-\phi_i)^2$, where $\gamma_1$ is the strength of the conformational entropy and  the double-well bulk energy parameterized by $\gamma_2$.

\subsection{Non-dimensionalization}
For system \eqref{incomp-model-eq} with $\hat{\mu}_i$ defined in \eqref{hat-mu-i}, using characteristic scales of length $L_0$ and velocity $V_0$, we nondimensionalize the physical variables and parameters as follows:
\begin{align}\label{Non-dim-eq}
	\begin{split}
		&\bx^{*} = \frac{\bx }{L_0},\quad
		\bv^{*} = \frac{\bv}{V_0}, \quad
		t^* = \frac{tV_0}{L_0},\quad
		\rho^* = \frac{\rho }{\rho_0}, \quad
		p^* = \frac{p}{\rho_0 V_0^2},\\
		&\eta^* = \frac{\eta }{\rho_0 V_0 L_0},\quad
		\bsig_i^* = \frac{\bsig_i}{\rho_0 V_0^2},	\quad
		\epsilon_2^* = \frac{\epsilon_2 }{\rho_0 V^2_0 L_0^2},\quad
		\alpha_i^* = \frac{  \alpha_{i}  }{\rho_0 V_0 L^2_0}, \\
		&\hat{\mu}_i^* = \frac{\hat{\mu}_i}{\rho_0 V_0^2},\quad
		M^* = \frac{M \rho_0 V_0}{L_0},\quad
		\bR_i^* = \frac{\bR_i}{L_0},\quad R^*_m = \frac{R_m}{L_0},
	\end{split}
\end{align}
The governing equations \eqref{incomp-model-eq} in non-dimensional forms are given as follows.
\begin{numcases}{}
	\frac{\d \bR^*_i}{\d t^*} = \bV^*_i(t^*),\quad  i = 1,\cdots, N,\\
	\rho^* \frac{\p \bv^*}{\p t^*} =-\nabla^* p^*+\eta^* \Delta^* \bv^* + \nabla^* \cdot\left( \sum_{i =1}^N\phi_{ i }\bsig^*_i \right) + \sum_{i=1}^N [\bF^*_{i, \ex}+\bF^*_{i,b}+\bF^*_{i, \drag}]\frac{\phi_i}{\int_{\Omega} \phi_i \d\bx}   \notag\\
	\qquad\quad -\sum_{i=1}^N \phi_i \nabla \hat{\mu}^*_i ,\\
	\nabla^* \cdot \bv^* = 0,\\
	\frac{\p \phi_i}{\p t^*} + \nabla^* \cdot (\phi_i \bv^*) = \nabla^* \cdot (M_i^* \nabla^* \hat{\mu}^*_i),\\
	\phi_i \bD^* = \b0,
\end{numcases}
where $\nabla^* = \left(\frac{\p }{\p x^*},   \frac{\p}{\p y^*} \right)^T$ with $\bx^* = (x^*, y^*)^T$, $\bD^* = \frac{ 1}{2} \left(\nabla^* \bv^* + (\nabla^* \bv^*)^T\right)$, and
\begin{align}\label{Non-dim-eq2}
	\begin{split}
		&\bV^*_i = \frac{1}{\int_{\Omega} \phi_i \d\bx}\int_{\Omega} \phi_i \bv^*     \d\bx,\quad
		\bF^*_{i, \ex} = - \sum^N_{j  = 1}\nabla_{R_i^*}V^*(\bR_i^*, \bR_j^*),\quad \bF^*_{i, b} = - \sum^N_{i  = 1}\nabla_{R_i^*}V^*(\bR_i^*, \partial \Omega),\\
		&\bF^*_{i, \drag} = -\alpha^*_{i} \bV_i^*,\quad
		V^* = \begin{cases}
			\epsilon^*_2 [(\frac{R^*_m}{R^*_{ij}})^{12}-2(\frac{R^*_m}{R^*_{ij}})^6], \quad &R^*_{ij}\leq R^*_m, \\
			0, \quad  &R^*_{ij}>R^*_m.
		\end{cases}
		\\
		&V^*_b = \begin{cases}
			\epsilon^*_2 [(\frac{R^*_m}{R^*_{ij}})^{12}-2(\frac{R^*_m}{\dist(R^*_{i}, \partial \Omega)})^6], \quad &\dist(R^*_{i},\partial \Omega)\leq R^*_n, \\
			0, \quad  &\dist(R^*_{i}, \partial \Omega)>R^*_n.
		\end{cases}
	\end{split}
\end{align}
To facilitate discussions in subsequent sections,  hereafter we will drop the superscript $(\bullet)^*$ in the non-dimensionalal forms, with the understanding that all variables and parameters are appropriately nondimensionalized.

\section{Thermodynamically consistent  numerical approximation}\label{sec:SDSG}

The hybrid hydrodynamical model for passive embedded particles is thermodynamically consistent.
We would like to approximate the PDE system using thermodynamically consistent linear schemes following the energy quadratization approach coupled with
the linear stabilization strategy, which has proven to be particularly effective in enhancing stability while keeping the required accuracy of the schemes.

We introduce two new variables $q_2(t) = \sqrt{\fr \sum_{i \neq j = 1}^N V(\bR_i, \bR_j)+\sum_{i=1}^N V_b(\bR_i, \partial \Omega) + C_2}$ and $q_1(\bx, t) = \sqrt{ \sum_{i } g_i(\phi_i) + C_1 }$,
where $C_1$ and $C_2$ are  constants that ensure the radicand always  positive (in this paper, we set $C_1 = C_2 = 1$). Then, the total energy can be rewritten into
\begin{align}\label{EQ-energy-eq}
	\wh{E}(\bv, \phi_{i }, \bsig_i, q_1, q_2) =  \int_{\Omega} \left[ \frac{\rho}{2} |\bv|^2+ q_1^2 + \sum_{i=1}^N \left(\frac{\gamma_1}{2} |\nabla \phi_{i} |^2 + \frac{\varepsilon}{2} |\bsig_i|^2 + \frac{s}{2}|\nabla \phi_{i, t}|^2 \right) \right] \d \bx
	+|q_2|^2 - C_{e},
\end{align}
where $C_{e} = C_1|\Omega| + C_2$ and $|\Omega|$ denotes the area of domain $\Omega$.
By taking time derivatives of   $q_1(\bx, t)$ and $q_2(t)$  with respect to $t$, we rewrite the governing equations in the following system in   ($\bv,  \phi_{i }, \bsig_i, q_1, q_2$):
\begin{subequations}\label{EQ-incmp-model}
	\begin{numcases}{}
		\frac{\d q_2}{\d t} = \fr \sum_{i} \frac{\p q_2}{\p \bR_i} \frac{\d \bR_i }{\d t}, \label{EQ-incmp-model-eq0}\\
		\rho \bv_t =-\nabla p+\eta \Delta \bv + \nabla \cdot\left( \sum_{i =1 }^N\phi_{ i }\bsig_i \right) + \sum_{i=1}^N [-q_2 \frac{\p q_2}{\p \bR_i }+\bF_{i, \drag}]\frac{\phi_i}{\int_{\Omega} \phi_i \d \bx} -\sum_{i=1}^N \phi_i \nabla \hat{\mu}_i ,\label{EQ-incmp-model-eq1}\\
		\nabla \cdot \bv = 0, \label{EQ-incmp-model-eq2}\\
		\phi_{i, t} + \nabla \cdot (\phi_i \bv) = M \Delta \hat{\mu}_i \label{EQ-incmp-model-eq3}\\
		\hat{\mu}_i =   2 q_1 h_i(\phi_{i })  - \gamma_1 \Delta \phi_{i } - s \Delta \p_t\phi_{i,t}, \label{EQ-incmp-model-eq4}\\
		\varepsilon \bsig_{i, t} - \phi_i \bD = {\bm 0},\label{EQ-incmp-model-eq5}\\
		q_{1, t} = \sum_{i=1}^N h_i(\phi_i) \phi_{i, t },\label{EQ-incmp-model-eq6}
	\end{numcases}
\end{subequations}
where $s \geq 0$ is a user supplied constant,  $ h_i(\phi_{i }) = \frac{g_i^{\prime}(\phi_{i } )}{  2 \sqrt{ \sum_{i } g_i(\phi_i) + C_1 }  }$, $\frac{\p q_2}{\p \bR_i} = \frac{\sum\limits_{j} \nabla_{\bR_i} V(\bR_i, \bR_j)+\sum\limits_{i=1}^N \nabla_{\bR_i} V_b(\bR_i, \partial \Omega)}{\sqrt{\fr \sum\limits_{i \neq j = 1}^N V(\bR_i, \bR_j)+\sum\limits_{i=1}^N V_b (\bR_i, \partial \Omega) + C_2}}$ and $\bsig_i = \left[ \begin{matrix}
	{ \sigma_{i}}_{1 1}& { \sigma_{i}}_{1 2}\\
	{ \sigma_{i}}_{1 2}& -{ \sigma_{i}}_{1 1}
\end{matrix}\right]$. The boundary conditions are given in  \eqref{incomp-BCs}  and the
initial conditions  are given by
\begin{align}
	&\bv(\bx, 0) = \bv_0(\bx),\quad
	\phi(\bx, 0) = \phi_0(\bx),\quad
	\bsig_i(\bx, 0) = \bsig_{i, 0}(\bx),\\
	&q_1(\bx, 0) = \sqrt{\sum_{i }g_i(\phi_0) +C_1},\quad
	q_2(0) = \sqrt{\fr \sum_{i \neq j = 1}^N V(\bR_i(0), \bR_j(0))+\sum_{i=1}^N V_b(\bR_i(0),\partial \Omega) + C_2}.
\end{align}
The above initial-boundary value problem is equivalent to \eqref{relax-model-eq}. Taking the inner product of \eqref{EQ-incmp-model-eq0} with $2 q_2$ \eqref{EQ-incmp-model-eq1} with $\bv$, \eqref{EQ-incmp-model-eq3} with $\hat{\mu}_i$, \eqref{EQ-incmp-model-eq4} with $\phi_{ i, t }$ and \eqref{EQ-incmp-model-eq6} with $2 q_1$, respectively, we obtain the following equivalent energy dissipation law
\begin{align}\label{active-energy-dissip}
	\dfrac{\d }{\d t} \wh{E}(\bv, \phi_i, \bsig_i, q_1, q_2) =  -\eta \| \nabla \bv \|^2 -  \sum_{i  = 1}^N \left[   M \|\nabla \hat{\mu}_i \|^2  + \alpha_i \bV_i^2 \right],
\end{align}
where the total energy is defined in \eqref{EQ-energy-eq}. We next focus on this reformulated equivalent system of equations, and present an unconditionally energy stable scheme for approximating this system.

Notice that the time and spatial discretization of the PDE system can be done independently \cite{Hong2020Energy}. To save space and simplify our notation, we only present semi-discrete schemes in time in the following.  The spatial discretization of the semi-discrete system is carried out using the finite difference method on staggered grids. For details on the spatial discretization, interested readers please refer to  \cite{Hong2020Energy}.  The spatial discretization described in \cite{Hong2020Energy} respects summation-by-parts so that the resultant fully discrete schemes are thermodynamically consistent.

\subsection{Numerical scheme and discrete energy stability}

Let $\tau$ be a time step size,  $ t_n = n \tau $ for $0 \leq n \leq N_t$ with $T = N_t \tau$, and $\psi^n$  the numerical approximation to  $\psi(\cdot, t)$ at $t = t_n$ for any function $\psi$.
Next, we  use the second-order Crank-Nicolson method to discretize   system \eqref{EQ-incmp-model} in time,  the resulting temporal semi-discrete scheme reads as follows:
\begin{sch}\label{incomp-2nd-cp-sch}
	\begin{subequations}
		\begin{numcases}{}
			\delta_t^+ q_2^n = \frac{1}{2}\sum_{i} \overline{\frac{\p q_2}{\p \bR_i}}^{\hf}\delta_t^+\bR_i^n,\label{incomp-2nd-cp-sch-eq0}\\
			\rho \delta_t^+ \bv^n = -\nabla p^{\hf} + \eta \Delta \bv^{\hf} + \nabla \cdot \left(\sum_{i=1}^N \ol{\phi}^{\hf}_i \bsig^{\hf}_i\right)  - \sum_{i}\ol{\phi}^{\hf}_i \nabla \mu_i^{\hf}\notag\\
			\qquad\quad+\sum_{i}\left[-q_2^{\half}\ol{\frac{\p q_2}{\p \bR_i}}^{\half} + \bF_{i, \drag}^{\half} \right]\frac{\ol{\phi}_i^{\hf}}{\int_{\Omega}\ol{\phi}^{\hf}_i\d \bx},\label{incomp-2nd-cp-sch-eq1}\\
			\nabla \cdot \bv^{\hf} = 0,\label{incomp-2nd-cp-sch-eq2}\\
			\delta_t^+ \phi_i^n + \nabla \cdot (\ol{\phi}_i^{\half} \bv^{\half}) = M \Delta \hat{\mu}_i^{\hf},\label{incomp-2nd-cp-sch-eq3}\\
			\hat{\mu}_i^{\hf} = 2 q^{\hf}_1 \ol{h}_i^{\hf} - \gamma_1 \Delta \phi_i^{\hf} - S \Delta (\phi_i^{n+1} - 2 \phi_i^n + \phi_i^{n - 1}),\label{incomp-2nd-cp-sch-eq4}\\
			\delta_t^+ q_1^n = \sum_{i} \ol{h}_i^{\hf} \delta_t^+ \phi_i^n,\\
			\varepsilon \delta_t^+\bsig_i^n - \ol{\phi}_i^{\hf} \bD^{\hf} = 0.\label{incomp-2nd-cp-sch-eq5}
		\end{numcases}
	\end{subequations}
	where $S = s/\tau^2$, $\ol{h}^{\half}_i  = h_i(\ol{\phi}_i^{\half})$, $\delta_t^+(\bullet)^n = ((\bullet)^{n+1} - (\bullet)^n)/\tau$ and $\ol{(\bullet)}^{n+\fr} = (3 (\bullet)^n - (\bullet)^{n-1})/2$.
\end{sch}
By taking the continuous $L^2$ inner product of \eqref{incomp-2nd-cp-sch-eq3} with $1$, it yields
\begin{align}
	(\phi^{n+1}_i, 1) = (\phi_i^n, 1) = \cdots =(\phi_i^0, 1),\quad
	\forall\; i \in \{1, 2, \cdots, N\}.
\end{align}
So,   scheme \eqref{incomp-2nd-cp-sch} preserves the volume of the immersed particles at the discrete level. In addition, we have the following theorem for the total energy.
\begin{thm}
	Scheme \eqref{incomp-2nd-cp-sch} is unconditionally energy stable  satisfying the following discrete energy dissipation law
	\begin{align}
		\frac{1}{\tau} (\wh{E}^{n + 1} - \wh{E}^n + \wt{E}^{n})  = - \eta\|\nabla \bv^{\half}\|_h^2-\sum_{i  = 1}^N \left[   M \|\nabla \hat{\mu}^{\half}_i \|^2  + \alpha_i |\bV^{\half}_i|^2 \right],
	\end{align}
	where
	\begin{align*}
		&\wh{E}^n = \frac{\rho}{2} \|\bv^n\|^2 + \|q_1^n\|^2 + \sum_{i = 1}^N \left( \frac{\gamma_1}{2}\|\nabla_h \phi^n_i\|^2  + \frac{\varepsilon}{2}\|\bsig^n_i\|^2 + \frac{S}{2}\|\nabla \phi_i^n - \nabla \phi_i^{n - 1}\|^2 \right) + |q^n_2|^2 - C_{e},\\
		&\wt{E}^{n} =  \frac{S}{2 }\sum_{i=1}^N  \|\nabla \phi_i^{n+1} - 2 \nabla\phi_i^n  + \nabla \phi_i^{n - 1} \|^2.
	\end{align*}
\end{thm}
\begin{proof}
	Taking the inner  product of \eqref{incomp-2nd-cp-sch-eq0} with $2 q_2^{\hf}$ and \eqref{incomp-2nd-cp-sch-eq1} with $\bv^{\hf}$, respectively, and adding the results, we obtain
	\begin{align}
		\frac{\rho}{2}\delta_t^+\|\bv^n\|^2 + \delta_t^+ (q_2^n)^2 = -\eta\|\nabla \bv^{\hf}\|^2 - \sum_{i}[\frac{\varepsilon}{2}\delta_t^+\|\bsig_i^n\|^2 +\alpha_i  \|\bV_i^{\hf}\|^2 - (\mu_i^{\hf}, \nabla \cdot( \ol{\phi}_i^{\hf} \bv^{\hf}))].
	\end{align}
	Similarly, taking the inner product of \eqref{incomp-2nd-cp-sch-eq3} with $\hat{\mu}_i^{\hf}$, \eqref{incomp-2nd-cp-sch-eq4} with $\delta_t^+ \phi_i^n$ and \eqref{incomp-2nd-cp-sch-eq5} with $2 q_1^{\hf}$, respectively, and applying identity $2(a - b)(a - 2b + c) = |a - b|^2 - |b - c|^2 + |a - 2 b + c|^2$, we have
	\begin{align}
		\delta_t^+ \|q_1^n\|^2 + \sum_{i}\left[\delta_t^+(\frac{\gamma_1}{2}\|\nabla \phi_i^n\|^2+ \frac{S}{2}\|\nabla \phi_i^n - \nabla \phi_i^{n - 1}\|^2 ) + \frac{S}{2\tau} \|\nabla \phi_i^{n+1} - \nabla \phi_i^n + \nabla \phi_i^{n-1}\|^2\right] = -\sum_i \|\nabla \hat{\mu}_i^{\hf}\|^2.
	\end{align}
	Putting the above results together, we arrive at
	\begin{align}
		&\delta_t^+ \left(\frac{\rho}{2} \|\bv^n\|^2 + \|q_1^n\|^2 + \sum_{i = 1}^N \left( \frac{\gamma_1}{2}\|\nabla_h \phi^n_i\|^2  + \frac{\varepsilon}{2}\|\bsig^n_i\|^2 + \frac{S}{2}\|\nabla \phi_i^n - \nabla \phi_i^{n - 1}\|^2 \right) + |q^n_2|^2 \right) \\
		&+ \sum_{i=1}^N \frac{S}{2 \tau} \|\nabla \phi_i^{n+1} - \nabla \phi_i^n + \nabla \phi_i^{n-1}\|^2 = - \eta\|\nabla \bv^{\half}\|_h^2-\sum_{i  = 1}^N \left[   M \|\nabla \hat{\mu}^{\half}_i \|^2  + \alpha_i |\bV^{\half}_i|^2 \right].
	\end{align}
	This implies the desired conclusion and completes the proof.
\end{proof}

\begin{remark}
	It is straightforward to construct another second-order and fully-coupled scheme  by using the second-order backward  differentiation formula (BDF2), which satisfies
	\begin{align}
		\frac{1}{\tau} (\wh{E}_h^{n + 1} - \wh{E}_h^n + \wt{E}_h^{n})  = - \eta\|\nabla_h \bv^{n+1}\|_h^2-\sum_{i  = 1}^N \left[   M \|\nabla_h \hat{\mu}^{n+1}_i \|_h^2  + \alpha_i |\bV^{n+1}_i|^2 \right],
	\end{align}
	where
	\begin{align*}
		\wh{E}^n_h &= \frac{\rho}{2}  \frac{\|\bv^{n}\|_h^2 + \|2\bv^n - \bv^{n-1}\|_h^2}{2} + \frac{\|q_1^n\|_h^2 + \|2q_1^n - q_1^{n-1}\|_h^2}{2} + \sum_{i=1}^N \Big( \frac{\gamma_1}{2} \frac{\|\nabla_h \phi_i^{n}\|_h^2 + \|2 \nabla_h \phi_i^n - \nabla_h \phi_i^{n-1}\|_h^2}{2} \\
		&+ \frac{\varepsilon}{2} \frac{\|\bsig_i^{n}\|_h^2 + \|2\bsig_i^n - \bsig_i^{n-1}\|_h^2}{2} + \frac{S}{2}\|\nabla_h \phi_i^n - \nabla_h \phi_i^{n - 1}\|_h^2  \Big) + \frac{|q_2^n|^2 + |2q_2^n - q_2^{n-1}|^2}{2} - C_{e},\\
		\wt{E}_h^{n} &= \frac{\rho}{4}\|\bv^{n+1} - 2\bv^n + \bv^{n-1}\|_h^2 + \frac{1}{2}\|q_1^{n+1} - 2q_1^{n} + q_1^{n-1}\|^2_h + \sum_{i=1}^N \Big( \frac{\gamma_1}{4}\|\nabla_h\phi_i^{n+1} -2 \nabla_h\phi_i^n + \nabla_h \phi_i^{n-1}\|_h^2\\
		&+\frac{\varepsilon}{4}\|\bsig_{i}^{n+1} - 2\bsig_{i}^n + \bsig_{i}^{n-1}\|_h^2 + S \|\nabla_h \phi_i^{n+1} - 2\nabla_h \phi_i^n + \nabla_h \phi_i^{n-1}\|_h^2 \Big) + \frac{1}{2}|q_2^{n+1} - 2 q_2^n + q_2^{n-1}|^2.
	\end{align*}
	The proof of energy stability is similar to that in the theorem, in which following identities $2a(3a - 4b + c) = |a|^2 + |2a -b |^2 - |b|^2 - |2 b - c|^2 + |a - 2b + c|^2$ and $(a - 2b + c)(3a -4b +c) = |a-b|^2 - |b - c|^2 + 2|a - 2b + c|^2$ are required.
	The details are omitted for simplicity.
\end{remark}
Scheme \eqref{incomp-2nd-cp-sch} is a linear, but fully coupled scheme. We next present a decoupled one.

\subsection{A decoupled, second-order and structure-preserving numerical scheme }

In this subsection, we develop a second-order, linear, fully decoupled thermodynamically consistent numerical scheme for this model.
Firstly, we introduce
$q_2(t) = \sqrt{\fr \sum_{i \neq j = 1}^N V(\bR_i, \bR_j) +\sum_{i=1}^N V_b(\bR_i, \partial \Omega)+ C_2}$ and $r(t) = \sqrt{\int_{\Omega} \sum_{i} g_i(\phi_{ i }) \d \bx + C_1}$.
where $C_1$ and $C_2$ are constants that guarantee the radicand always positive.  Subsequently, we introduce a variable $s_2(t)$ such that  the original system can be recast into the following equivalent formulation.	
\begin{subequations}\label{decoupled-model}
	\begin{numcases}{}
		\frac{\d q_2}{\d t} = \frac{s_2}{2} \sum_{i} \frac{\p q_2}{\p \bR_i} \frac{\d \bR_i }{\d t}, \label{decoupled-model-eq0}\\
		\rho \wt{\bv}_t =-\nabla p+\eta \Delta  \wt{\bv} + s_2 \nabla \cdot\left( \sum_{i =1 }^N\phi_{ i }\bsig_i \right) + \sum_{i=1}^N [-s_2\frac{q_2  \p q_2}{\p \bR_i }+\bF_{i, \drag}]\frac{\phi_i}{\int_{\Omega} \phi_i \d \bx} \notag\\
		\qquad - s_2\sum_{i=1}^N \phi_i \nabla \hat{\mu}_i ,\label{decoupled-model-eq1}\\
		\bv = \wt{\bv} - \varepsilon_2 \nabla p_t, \label{decoupled-model-eq2}\\
		\nabla \cdot \bv = 0, \label{decoupled-model-eq3}\\
		\phi_{i, t} + s_2\nabla \cdot (\phi_i \wt\bv) = M \Delta \hat{\mu}_i \label{decoupled-model-eq4}\\
		\hat{\mu}_i =   2 r H_i(\phi_{i })  - \gamma_1 \Delta \phi_{i } - s \Delta \p_t\phi_{i,t}, \label{decoupled-model-eq5}\\
		r_{ t} = \sum_{i=1}^N H_i(\phi_i) \phi_{i, t },\label{decoupled-model-eq6}\\
		\varepsilon \bsig_{i, t} - s_2 \phi_i \bD = {\bm 0},\label{decoupled-model-eq7}\\
		\frac{\d s_2}{\d t} = \sum_{i=1}^N \bigg[-q_2 \frac{\p q_2}{\p \bR_i} \bV_i + \big(-\nabla \cdot(\phi_i \bsig_{i}) + q_2 \frac{\p q_2}{\p \bR_i} \frac{\phi_{ i }}{\int_{\Omega}\phi_{ i }\d \bx} + \phi_i \nabla \hat{\mu}, \wt\bv\big) \notag \\
		\qquad+ (\nabla\cdot (\phi_i \bv),\hat{\mu}_i) - (\phi_{ i }\bD, \bsig_{i})\bigg],\label{decoupled-model-eq8}
	\end{numcases}
	where $H_i(\phi_i) = \frac{g^{\prime}_i(\phi_{ i })}{\sqrt{\int_{\Omega}\sum_{i} g_i(\phi_i)\d \bx + C_1}}$. The boundary conditions are given in  \eqref{incomp-BCs} or periodic, the
	initial conditions  are the same as in the previous case and
	\begin{align}
		r(0) = \sqrt{\int_{\Omega} \sum_{i} g_i(\phi^0_{ i }) \d \bx + C_1},\quad s_2(0) = 1.
	\end{align}
\end{subequations}
\begin{remark}
	When   $s_2 \equiv 1$ ,  system \eqref{decoupled-model} reduces to the original system.
\end{remark}

An unconditionally energy stable scheme for  equation \eqref{decoupled-model} is given in the following.
\begin{sch}\label{decoupled-2nd}
	\begin{subequations}\label{decoupled-2nd-eq}
		\begin{numcases}{}
			\delta_t^+ q_2^n = \frac{s^{n+\fr}_2}{2} \sum_{i}\ol{\frac{\p q_2}{\p \bR_i}}^{\half} \ol{\bV_i}^{\half}, \label{decoupled-2nd-eq0}\\
			\rho \frac{\wt{\bv}^{n+1} - \bv^n}{\tau} = - \nabla p^n + \eta \Delta \wt{\bv}^{n+\fr} + s_2^{\half} \nabla \cdot (\sum_{i} \ol{\phi}_i^{\half} \ol\bsig_i^{\half}) +  \notag \\
			s_2^{\half} \sum_{i} -\ol{q}^{\half}_2\ol{\frac{\p q_2}{\p \bR_i}}^{\half} \frac{\ol{\phi}_i^{\half}}{\int_{\Omega} \ol{\phi}_i^{\half} \d \bx} + \sum_{i}  \ol{\bF}^{\half}_{i, \drag} \frac{\ol{\phi}_i^{\half}}{\int_{\Omega} \ol{\phi}_i^{\half} \d \bx} - s_2^{\half} \sum_{i} \ol{\phi}_i^{\half }\nabla \ol{\hat{\mu}}_i^{\half}, \label{decoupled-2nd-eq1}
		\end{numcases}
	\end{subequations}
\end{sch}

\begin{subequations}
	\begin{numcases}{}
		\bv^{n+1} = \wt{\bv}^{n+1} - \varepsilon_2 (\nabla p^{n+1} - \nabla p^n)/\tau, \label{decoupled-2nd-eq8}\\
		\nabla \cdot \bv = 0, \label{decoupled-2nd-eq2}\\
		\delta_t^+ \phi_{ i }^n + s_2^{\half} \nabla \cdot (\ol {\phi}_i^{\half} \ol{\bv}^{\half}) = M \Delta \hat{\mu}^{\half}_i, \label{decoupled-2nd-eq3}\\
		\hat{\mu}^{\half}_i = r^{\half} \ol{H}_i - \gamma_1 \Delta \phi_{i }^{\half} - S \Delta (\phi_i^{n + 1} - 2\phi_i^n + \phi_i^{n - 1}), \label{decoupled-2nd-eq4}\\
		\delta_t^+ r^n = \fr \sum_{i  } (\ol{H}_i , \delta_t^+ \phi^n_{ i }), \label{decoupled-2nd-eq5}\\
		\varepsilon \delta_t^+ \bsig_i^n - s_2^{\half} \ol{\phi}_i^{\half} \ol\bD^{\half} = 0, \label{decoupled-2nd-eq6}\\
		\delta^+_t s_2^n = \sum_{i} \bigg[-q^{\half }_2 \ol{\frac{\p q_2}{\p \bR_i}}^{\half} \ol{\bV_i}^{\half} -(\nabla \cdot (\ol{\phi}_i \ol{\bsig}_i^{\half}), \wt\bv^{\half}) +\notag\\
		(\ol{q}_2^{\half}\ol{\frac{\p q_2}{\p \bR_i}}^{\half} \frac{\ol{\phi}_i^{\half}}{\int_{\Omega} \ol{\phi}_i^{\half}\d \bx}, \wt{\bv}^{\half} ) + (\ol{\phi}_i^{\half} \nabla \ol{\hat{\mu}}_i^{\half} , \wt{\bv}^{\half}) +\notag\\
		(\nabla \cdot (\ol\phi_i^{\half} \ol{\bv}^{\half}), \hat{\mu}_i^{\half}) - (\ol{\phi}_i \ol\bD^{\half}, \bsig_i^{\half})
		\bigg],\label{decoupled-2nd-eq7}
	\end{numcases}
\end{subequations}
where $\ol{H}_i = H_i(\ol{\phi}^{\half}_i)$, $\wt{\bv}^{\half} = (\wt{\bv}^{n+1} + {\bv}^n)/2$, $S = s / \tau^2$, $\varepsilon_2 = \tau^2 / (2 \rho)$.  This scheme allows $\bv, \phi_i, p, \bsig_i$ fully decoupled.

\begin{thm}
	Scheme \eqref{decoupled-2nd} is unconditionally energy stable satisfying the following discrete energy dissipation law
	\begin{align}
		\frac{1}{\tau} \left( \wh{E}^{n+1} - \wh{E}^n + \wt{E}^n \right) + \eta\|\nabla \wt\bv^{\half}\|^2 + \sum_{i}\left[ M \|\nabla \hat{\mu}_i^{\half}\|^2 + \alpha_i (\wt{\bV}^{\half}_i)^2\right] = 0,
	\end{align}
	where $\wt{\bV}_i^{\half} = \frac{1}{\int_{\Omega} \ol{\phi}^{\half}_i \d \bx} \int_{\Omega} \ol{\phi}^{\half}_i \wt{\bv}^{\half}\d \bx$ and
	\begin{align}
		\wh{E}^{n} &= \frac{\rho}{2}\|\bv^n\|^2 + \frac{\tau^2}{8 \rho} \|\nabla p^n\|^2 + \sum_{i} ( \frac{\gamma_1}{2}\|\nabla \phi_i^n\|^2 + \frac{\varepsilon}{2}\|\bsig_i^n\|^2 + \frac{S}{2}\|\nabla \phi_i^n - \nabla \phi_i^{n-1}\|^2 ) \label{DSAV-energy-eq}\\
		&+ |r^n|^2+|q^n_2|^2 + \frac{|s_2^n|^2}{2} - C_{s},\quad C_{s} = C_1 + C_2,\\
		\wt{E}^n &= \frac{S}{2} \sum_{i} \|\nabla \phi_i^{n+1} - 2\nabla \phi_i^n + \nabla \phi_i^{n-1}\|^2.
	\end{align}
\end{thm}

\begin{proof}
	Multiplying \eqref{decoupled-2nd-eq0}  with $2 q_2^{\half}$ and summing over $i$, computing the discrete inner product of \eqref{decoupled-2nd-eq1}, \eqref{decoupled-2nd-eq3}, \eqref{decoupled-2nd-eq4}, \eqref{decoupled-2nd-eq5} and  \eqref{decoupled-2nd-eq6}
	with $\wt{\bv}^{\half}$,  $\hat{\mu}_{ i }^{n+\fr}$, $-\delta_t^+ \phi_i^n$, $2 r^{\half}$, $\bsig_i^{\half}$, respectively, multiplying \eqref{decoupled-2nd-eq7} with $s_2^{\half}$,
	taking the discrete inner product of \eqref{decoupled-2nd-eq8} with $\bv^{n + 1}$ and $\frac{\tau}{2 \rho}\nabla p^n$, respectively, and putting the above results together, we obtain the desired result with the aide of identities $2(a - b) (a - 2 b + c) = |a - b|^2 - |b - c|^2 + |a - 2 b + c|^2$ and $2a(a - b) = |a|^2 - |b|^2 + |a- b|^2$.
\end{proof}

\begin{remark}
	We can develop a second-order and decoupled time-marching scheme using BDF2 and choosing $\varepsilon_2 = 2\tau^2/(3\rho)$ in \eqref{decoupled-model-eq2}. To save space, we only present the corresponding semi-discrete energy dissipation law next
	\begin{align}
		\frac{1}{\tau} (\wh{E}^{n + 1} - \wh{E}^n + \wt{E}^{n})  = - \eta\|\nabla \wt{\bv}^{n+1}\|^2-\sum_{i  = 1}^N \left[   M \|\nabla \hat{\mu}^{n+1}_i \|^2  + \alpha_i |\wt{\bV}^{n+1}_i|^2 \right],
	\end{align}
	where $\wt{\bV}_i^{n+1} = \frac{1}{\int_{\Omega} \ol{\phi}^{n+1}_i \d \bx} \int_{\Omega} \ol{\phi}_i^{n+1} \wt{\bv}^{n+1}\d \bx$ with $\ol{\phi}_i^{n+1} = 2\phi_i^n - \phi_i^{n-1}$, and
	\begin{align*}
		\wh{E}^n &= \frac{\rho}{2}  \frac{\|\bv^{n}\|^2 + \|2\bv^n - \bv^{n-1}\|^2}{2} +\frac{\tau^2}{3 \rho}\|\nabla p^{n}\|^2 + \sum_{i=1}^N \Big( \frac{\gamma_1}{2} \frac{\|\nabla_h \phi_i^{n}\|^2 + \|2 \nabla \phi_i^n - \nabla \phi_i^{n-1}\|^2}{2}\\
		&+ \frac{\varepsilon}{2} \frac{\|\bsig_i^{n}\|_h^2 + \|2\bsig_i^n - \bsig_i^{n-1}\|^2}{2} + \frac{S}{2}\|\nabla \phi_i^n - \nabla \phi_i^{n - 1}\|^2  \Big) + \frac{|r^n|^2 + |2r^n - r^{n-1}|^2}{2}\\
		&+ \frac{\|q_2^n\|_h^2 + \|2q_2^n - q_2^{n-1}\|^2}{2}  + \fr \frac{|s_2^{n}|^2 + |2 s_2^{n} - s_2^{n-1}|^2}{2} - C_s,\\
		\wt{E}^{n} &= \frac{\rho}{4}\|\bv^{n+1} - 2\bv^n + \bv^{n-1}\|^2 + \frac{3 \rho}{4} \| \bv^{n+1} - \wt{\bv}^{n+1}\|^2   + \sum_{i=1}^N \Big( \frac{\gamma_1}{4}\|\nabla \phi_i^{n+1} -2 \nabla \phi_i^n + \nabla \phi_i^{n-1}\|^2\\
		&+\frac{\varepsilon}{4}\|\bsig_{i}^{n+1} - 2\bsig_{i}^n + \bsig_{i}^{n-1}\|_h^2 + S \|\nabla \phi_i^{n+1} - 2\nabla \phi_i^n + \nabla \phi_i^{n-1}\|^2 \Big) + \fr (|r^{n+1} - 2r^n + r^{n-1}|^2) \\
		&+ \frac{1}{2}|q_2^{n+1} - 2 q_2^n + q_2^{n-1}|^2 + \frac{1}{4}|s_2^{n+1} - 2 s_2^{n} + s_2^{n-1}|^2.
	\end{align*}
	This shows that the decoupled scheme based on BDF2 is also energy stable.
\end{remark}

Next, we  discuss how to implement decoupled scheme \eqref{decoupled-2nd} efficiently. Combining equations \eqref{decoupled-2nd-eq3}-\eqref{decoupled-2nd-eq5} yields
\ben\label{phi-i-leq}
\bea{l}
r^{\half} = \fr \sum_{i} (\ol{H}_i, \phi_i^{\half }) + e^n,\quad e^n = r^n - \fr \sum_{i} (\ol{H}_i, \phi_i^{n}),\\[0.2cm]
\phi_i^{\half } - \frac{\cA^{-1}d_i^n}{2}\sum_{i}(\ol{H}_i, \phi_i^{\half}) = \cA^{-1}(f_i)^n_1 + s_2^{n+\fr} \cA^{-1} (f_i)^n_2,\\
\eea
\een
where $\cA^{-1}$ represents the inverse of $\cA$ with $\cA = 2 / \tau + (\gamma_1 + 2 S)M \Delta^2$, $d_i^n = M \Delta \ol{H}_i$, $(f_i)^n_1 = \frac{2}{\tau}\phi_i^n + d_i^n e^n + SM \Delta^2(3\phi_i^n - \phi_i^{n-1})$ and $(f_i)^n_2 = - \nabla \cdot (\ol{\phi}_i^{\half} \ol{\bv}^{\half})$. Taking advantage of the fact that $s_2^{\half}$ is a scalar variable, we introduce the pairs of field functions $((\phi_i)^{\half}_{1}, (\phi_i)^{\half}_{2})$, $((\hat\mu_i)^{\half}_{1}, (\hat\mu_i)^{\half}_{2})$ and $(r_1^{\half}, r_2^{\half})$ as the solutions to the following system:
\balc\label{phi-mu-eq}
(\phi_{i})_1^{\half } - \frac{\cA^{-1} d_i^n}{2}\sum_{i=1}^N(\ol{H}_i, (\phi_{i})_1^{\half }) = \cA^{-1}(f_i)^n_1,\\[0.1cm]
(\phi_{i})_2^{\half } - \frac{\cA^{-1}d_i^n}{2}\sum_{i=1}^N(\ol{H}_i, (\phi_{i})_2^{\half }) = \cA^{-1}(f_i)^n_2,\\[0.1cm]
(\hat{\mu}_{i})^{\half}_1 = \ol{H}_i r_1^{\half} - (\gamma_1 + 2S)\Delta (\phi_{i})_1^{\half }  + S \Delta (3 \phi_i^n - \phi_{ i }^n),\\[0.1cm]
(\hat{\mu}_{i})^{\half}_2 = \ol{H}_i r_2^{\half}  - (\gamma_1 + 2 S)\Delta (\phi_{i})_2^{\half },\\[0.1cm]
r_1^{\half} = e^n + \fr \sum_{i=1}^N (\ol{H}_i, \phi^{\half}_{i, 1}),\\[0.1cm]
r_2^{\half} = \fr \sum_{i=1}^N (\ol{H}_i, \phi^{\half}_{i, 2}).
\ealc
Given $s_2^{\half}$, $\phi_i^{\half} = (\phi_i)_1^{\half} + s_2^{\half }(\phi_i)_{ 2}^{\half} $, ${\hat\mu}_i^{\half} = (\hat \mu_i)^{\half}_{1} + s_2^{\half} (\hat\mu_{i})_2^{\half}$ and $r^{\half} = r_1^{\half} + s_2^{\half} r_2^{\half}$ satisfy the equations consisting of \eqref{decoupled-2nd-eq3}, \eqref{decoupled-2nd-eq4} and \eqref{decoupled-2nd-eq5}, respectively,   the boundary conditions defined in  \eqref{incomp-BCs} or periodic boundary conditions.

To solve for $(\phi_i)^{n+\frac{1}{2}}_j, j=1,2$, we show the linear system for $(\ol{H}_i, (\phi_i)_1^{\half})$  is uniquely solvable. Following the same argument,  the solvability of $(\ol{H}_i, (\phi_i)_2^{\half})$ can be established.  Define
\ben
\beta_{i} = \fr (\cA^{-1} M\Delta \ol{H}_i, \ol{H}_i).
\een
By \eqref{phi-mu-eq} and some simple calculations, the determinant of the coefficient matrix of the linear system for $(\ol{H}_i, (\phi_i)_1^{\half})$ is $1 - \sum_{i = 1}^N \beta_i$. Due to  $\cA^{-1} \Delta$ is a negative semi-definite operator, it is easy to deduce that $\beta_i \leq 0$, which implies the linear system for $(\ol{H}_i, (\phi_i)_1^{\half})$ is uniquely solvable.  Subsequently, plugging $(\ol{H}_i, (\phi_i)_1^{\half})$ and $(\ol{H}_i, (\phi_i)_2^{\half})$ into the first two equations of \eqref{phi-mu-eq}, we obtain $(\phi_{ i })^{\half}_1$ and $(\phi_{ i })^{\half}_2$ uniquely.

Next, in order to solve \eqref{decoupled-2nd-eq1} for $\wt{\bv}^{\half}$, we define two field functions $\wt{\bv}_1^{\half}$ and $\wt{\bv}_2^{\half}$ as the solutions to the following equations:
\balc\label{wt-v-eq}
\wt\bv_1^{\half} + \sum_{i=1}^N \cB^{-1}\xi_i (\ol\phi_i, \wt{\bv}^{\half}_1) = \cB^{-1}\bg_1^n, \\
\wt\bv_2^{\half} + \sum_{i=1}^N \cB^{-1} \xi_i (\ol\phi_i, \wt{\bv}^{\half}_2) = \cB^{-1}\bg_2^n,
\ealc
where $\cB^{-1}$ denotes the inverse of $\cB$ with $\cB = \frac{2 \rho}{\tau} - \eta \Delta$, $\xi_i = \frac{ \alpha_i \ol{\phi}_i}{(\int_{\Omega}\ol \phi_i \d \bx)^2}$, $\bg_1^n = \frac{2 \rho}{\tau} \bv^n - \nabla p^n$ and
\ben
\bg_2^n = \sum_{i=1}^N \big[ \nabla \cdot ( \ol{\phi}_i^{\half} \ol\bsig_i^{\half})    -\ol{q}^{\half}_2\ol{\frac{\p q_2}{\p \bR_i}}^{\half} \frac{\ol{\phi}_i^{\half}}{\int_{\Omega} \ol{\phi}_i^{\half} \d \bx} - \ol{\phi}_i^{\half }\nabla \ol{\hat{\mu}}_i^{\half} \big].
\een
Then, using the fact that $s_2^{\half}$ is a scalar, the solution to \eqref{decoupled-2nd-eq1} can be written as
\ben
\wt{\bv}^{\half} = \wt{\bv}_1^{\half}  + s_2^{\half} \wt{\bv}_2^{\half},\quad
\wt{\bv}_1|_{\p \Omega} = \wt{\bv}_2|_{\p \Omega}= 0\; \text{or}\; periodic.
\een
We note that procedures of computing $\wt{\bv}_1^{\half}$ and $\wt{\bv}_2^{\half}$ are similar to $(\phi_i)_1^{\half}$ and $(\phi_i)_2^{\half}$.
Similarly, for variables $\bsig_i^{\half}$ and $q_2^{\half}$ in system \eqref{decoupled-2nd-eq},
we use  $s_2^{\half}$ to split them in the following forms:
\balc\label{four-vari-eq}
\bsig_i^{\half} = ({\bsig_{i}})_1^{\half} + s_2^{\half} ({\bsig_{i}})_2^{\half},\\
q_2^{\half} = ({q_2})_1^{\half} + s_2^{\half} ({q_2})_2^{\half},
\ealc
where
\ben\label{q2-bsig-eq}
\bea{l}
(q_2)_1^{\half} = q_2^n,\quad
(q_2)_2^{\half} = \frac{\tau}{4} \sum_{i} \ol{\frac{\p q_2}{\p \bR_i}}^{\half} \ol{\bV}_i^{\half},\\[0.3cm]
(\bsig_{i})_1^{\half} = {\bsig_{i}}^n, \quad
(\bsig_{i})_2^{\half} = \frac{\tau}{2 \varepsilon} \ol{\phi}^{\half}_i \ol\bD^{\half}.
\eea
\een
Now, we are ready to determine $s_2^{\half}$.  Using the existing results above, we compute $s_2^{\half}$ from \eqref{decoupled-2nd-eq7} as follows
\ben\label{s2-eq}
\left(\frac{2}{\tau} - \theta_2\right) s_2^{\half} = \frac{2}{\tau} s_2^n + \theta_1,
\een
where
\ben
\bea{c}
\theta_k = \sum_{i} \Big[-(q_2)^{\half }_{k} \ol{\frac{\p q_2}{\p \bR_i}}^{\half} \ol{\bV_i}^{\half} -(\nabla \cdot (\ol{\phi}_i \ol{\bsig}_i^{\half}), \wt\bv_k^{\half}) + (\ol{q}_2^{\half}\ol{\frac{\p q_2}{\p \bR_i}}^{\half} \frac{\ol{\phi}_i^{\half}}{\int_{\Omega} \ol{\phi}_i^{\half}\d \bx}, \wt{\bv}_k^{\half} )\notag\\
+ (\ol{\phi}_i^{\half} \nabla \ol{\hat{\mu}}_i^{\half} , \wt{\bv}_k^{\half}) + (\nabla \cdot (\ol\phi_i^{\half} \ol{\bv}^{\half}), (\hat{\mu}_{i})_k^{\half}) - (\ol{\phi}_i \ol\bD^{\half}, (\bsig_{i})_k^{\half})
\Big],\quad
k = 1, 2.
\eea
\een
Once $s_2^{\half}$ is known, $\bv^{n+1}$ and $p^{n+1}$ are updated by \eqref{decoupled-2nd-eq8} and \eqref{decoupled-2nd-eq2}, respectively.
We summarize the   solution procedure in an algorithm given next.\\[0.1cm]

\begin{algorithm}[H]
	\caption{An efficient fully-decoupled, second-order and structure-preserving scheme}
	\textbf{Input}{ ($\bv^{n-1}$, $\phi_i^{n-1}$, $p^{n-1}$, $r^{n-1}$, $\bsig_{i}^{n-1}$, $q_2^{n-1}$) and ($\bv^{n}$, $\phi_i^n$, $p^n$, $r^n$, $\bsig_{i}^n$, $q_2^n$)}\\[0.2cm]
	\textbf{Output}($\bv^{n+1}$, $\phi_i^{n+1}$, $p^{n+1}$, $r^{n+1}$, $\bsig_{i}^{n+1}$, $q_2^n$)
	
	\hspace{0.01in}{}\\
	{\bf begin}\\[0.1cm]
	Solve $(\phi_i)^{\half}_k$ and $(\hat{\mu}_i)^{\half}_k$, $(k = 1, 2)$ by \eqref{phi-mu-eq}.
	
	\hspace{0.01in}{}\\
	Compute $\wt{\bv}^{\half}_k$, $(k = 1, 2)$ by \eqref{wt-v-eq}.
	
	\hspace{0.01in}{}\\
	Solve \eqref{q2-bsig-eq} for $(q_2)^{\half}_{k}$ and $(\bsig_{i})^{\half}_{k}$, $(k = 1, 2)$.
	
	\hspace{0.01in}{}\\
	Compute $s_2^{n+\fr}$ using \eqref{s2-eq}.
	
	\hspace{0.01in}{}\\
	Update $\bv^{n+1}$ and $p^{n+1}$ by \eqref{decoupled-2nd-eq8} and \eqref{decoupled-2nd-eq2}.\\[0.15cm]
	{\bf end}
\end{algorithm}

\begin{remark}
	Multiplying the first equation  \eqref{q2-bsig-eq} with $\frac{4}{\tau}(q_2)_2^{\half}$, taking the discrete inner product of the last equation of \eqref{q2-bsig-eq} with $\frac{2 \varepsilon}{\tau} (\bsig_{i})^{\half}_2$, \eqref{wt-v-eq} with $\wt{\bv}_2^{\half}$, \eqref{phi-mu-eq} with $(\hat{\mu}_{ i })^{\half}_2$, we deduce  $\theta_2 \leq 0$, yielding $\frac{2}{\tau} - \theta_2 > 0$. Thus, \eqref{s2-eq} is uniquely solvable.  In addition,
	the resulting linear system with only constant coefficient matrix can be solved at each time step using  fast algorithms  efficiently.
\end{remark}

\begin{remark}
	For the hydrodynamical model for active particles, the two newly developed schemes for passive particles can be readily extended to the ones that preserve the discrete energy dissipation rate. Although the energy may not be dissipative in the active matter system, the numerical schemes respect the energy production rate of the continuous model. In addition, it is easy to  prove theoretically that the discrete system is uniquely solvable with the solvability condition $(p, 1) = 0$. Due to space limitation, we omit the details of the proof and interested reader please refer  to \cite{Gong_SISC2018} for relevant details and the references therein.
\end{remark}

\begin{figure}[H]
	\centering
	\subfigure[Accuracy test for the coupled scheme \eqref{incomp-2nd-cp-sch}]{
		\includegraphics[width=0.45\textwidth,height=0.225\textwidth]{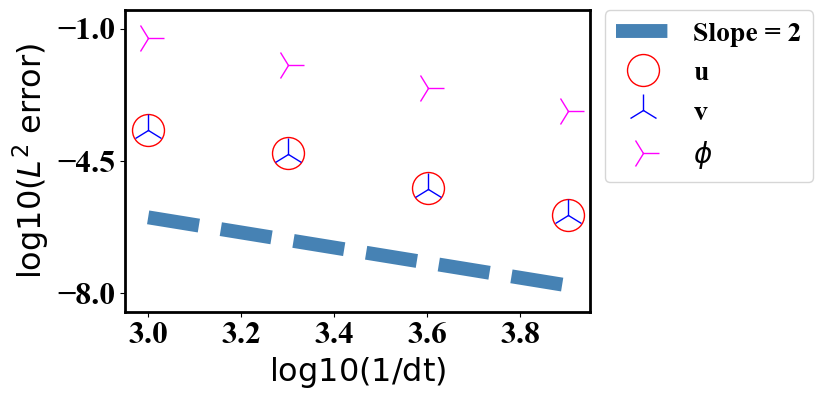}
		\includegraphics[width=0.45\textwidth,height=0.225\textwidth]{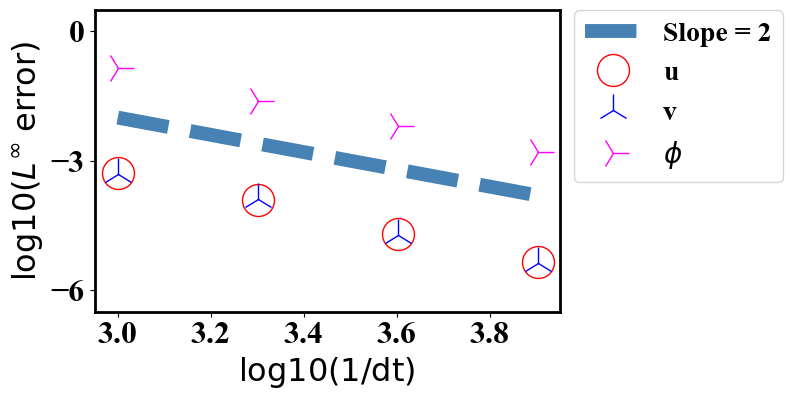}
	}
	\subfigure[Accuracy test for the decoupled scheme \eqref{decoupled-2nd}]{
		\includegraphics[width=0.45\textwidth,height=0.225\textwidth]{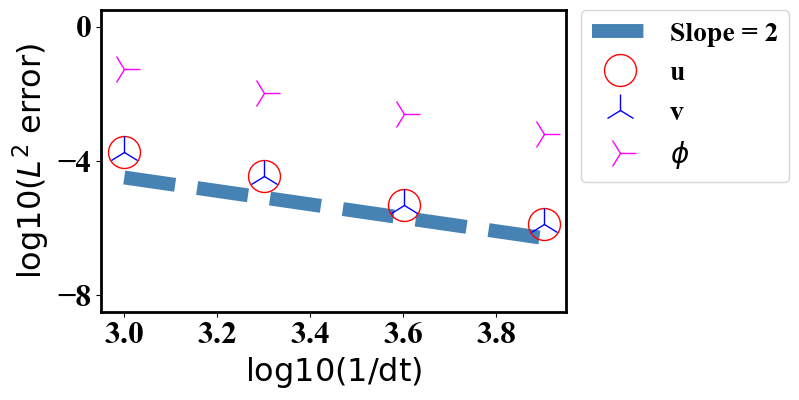}
		\includegraphics[width=0.45\textwidth,height=0.225\textwidth]{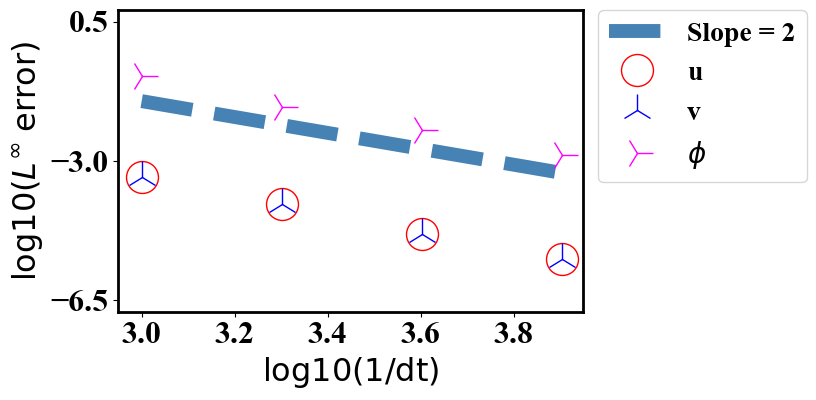}
	}
	\caption{{\bf Example \ref{eg:test-order-1}:} Mesh refinement test for the fully discrete coupled and decoupled scheme, respectively. Second order convergence rates for all variables in both time and space are confirmed. } \label{fig:test-order1}
\end{figure}

\section{Results and Discussion}
In this section, we first conduct the mesh refinement test in time and space to verify accuracy of the proposed schemes. Then we present several numerical examples of the fluid structure interaction to demonstrate the efficiency and usefulness  of  the new model and the developed structure-preserving numerical schemes. In all examples,  the double-well bulk energy for multiple particles is defined as $\sum_{i=1}^{k}g_i(\phi_{ i }) = \sum_{i=1}^{k}\gamma_2 \phi_i^2(1 - \phi_{ i })^2 + \Lambda \Pi_{i=1}^k\phi_i^2, k\geq 2$, where $\Lambda$ is a non-negative constant, and $C_1=C_2=1$.

\begin{example}[Convergence rate]\label{eg:test-order-1}
	We use one particle immersed in viscous fluid matrix in this example to demonstrate that the convergence rate of the proposed scheme is  second-order accurate both in time and space.  We use the following initial conditions
	\begin{align}
		\begin{cases}
			\bu(x, y, 0)  = \b0,\quad \\
			\phi_1(x, y, 0)  = \fr \left( 1 + \tanh \frac{r_1 - \sqrt{(x - x_1)^2 + (y - y_1)^2}}{\sqrt{2} \epsilon} \right),
		\end{cases}
	\end{align}
	where $\bu = (u, v)^T$.
	
	We use computational domain  $\Omega=[0, 1] \times [0, 1]$.
	The parameter values in the model and the initial condition of the particle are chosen as $(r_1, x_1, y_1) = (0.3, 0.5, 0.5)$, $\rho = 0.1$, $\eta = 10$, $\epsilon = 0.025$, $\gamma_1 = 0.25$, $\gamma_2 = 25$, $\alpha_1  = 50$, $S = \epsilon$, $M = 10^{-4}$ and   $\varepsilon = 0.001$, where $r_1$ is the radius of the particle.
	We choose  $N \times N$ spatial meshes, with $N = 8, 16, 32, 64, 128$, and  time step $ \tau =  10^{-3} \times \frac{1}{2^{k-1}}$, $k= 1, 2, \cdots 5$, respectively. The discrete $L^2$ and $L^{\infty}$
	errors  between the solution of the coarse mesh and that of the adjacent finer  mesh are calculated at $t = 0.1$. The numerical errors in mesh refinement tests are plotted in Figure \ref{fig:test-order1}, where we clearly observe that the expected second order accuracy in both time and space  for all variables are attained.
\end{example}

\begin{example}[Sensitivity test for the stabilizing parameter]\label{eg:ST}
	\begin{figure}[H]
		\centering
		\subfigure[The modified energy curve for $S=\epsilon$ with various time steps.]{
			\includegraphics[width=0.40\textwidth,height=0.25\textwidth]{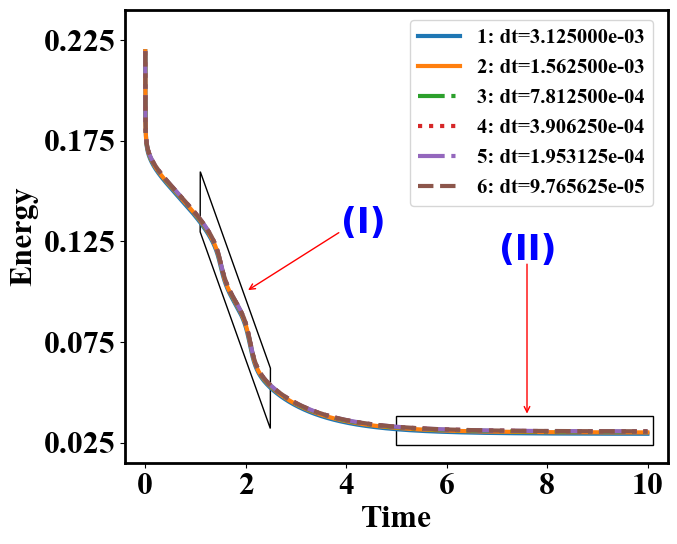}
		}\qquad
		\subfigure[The modified energy curve for $S = 0$ with various time steps.]{
			\includegraphics[width=0.40\textwidth,height=0.25\textwidth]{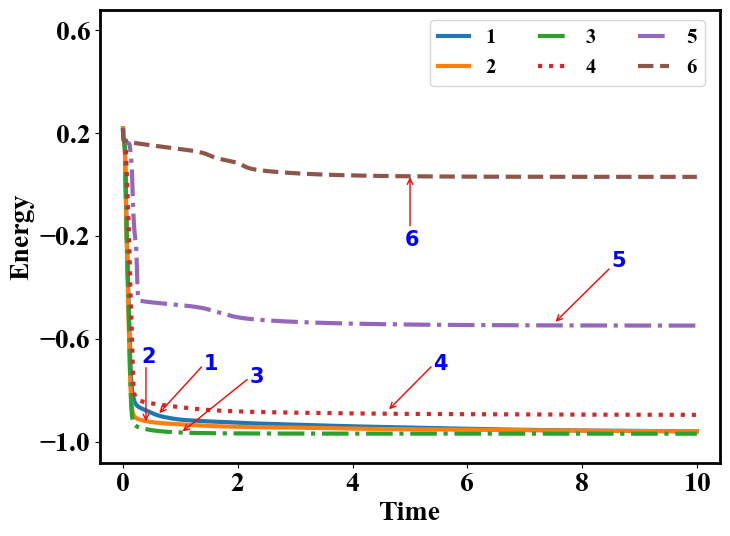}
		}\\
		\subfigure[Zoom-in  view of (a) for $1\leq t \leq 3$.]{
			\includegraphics[width=0.40\textwidth,height=0.25\textwidth]{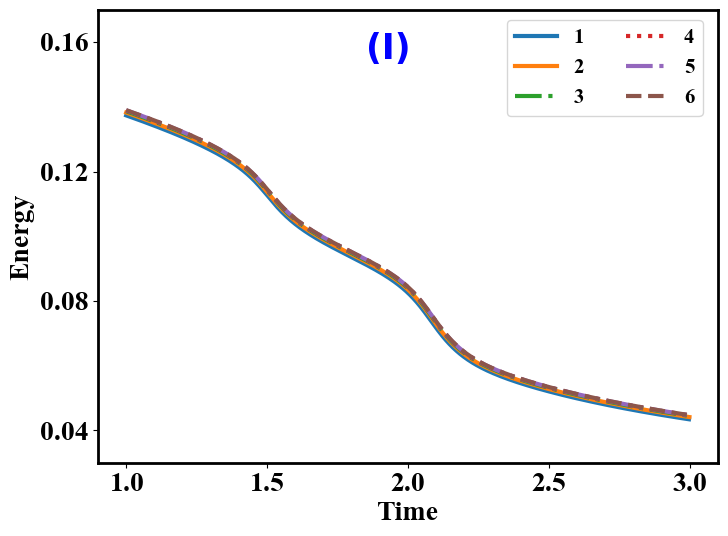}
		}\qquad
		\subfigure[Zoom-in view of (a) for $5 \leq t \leq 10$.]{
			\includegraphics[width=0.40\textwidth,height=0.25\textwidth]{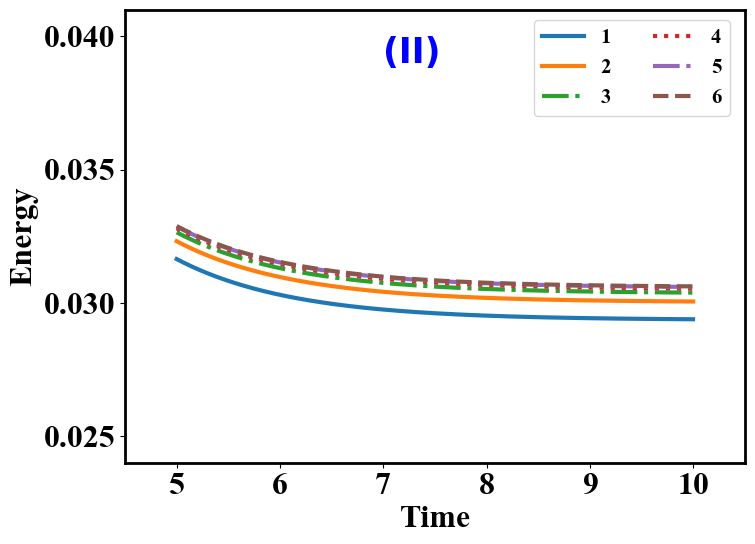}
		}
		\caption{{\bf Example \ref{eg:ST}:} The comparison of time evolution of the total energy with different time steps computed using coupled scheme \eqref{incomp-2nd-cp-sch} with $S = \epsilon$ and $S = 0$, respectively, where the mobility parameter is $M=10^{-3}$.
			The energy curves computed using $S = \epsilon$ 	agree  well when using relative larger time steps; however, the corresponding energy curves computed using $S = 0$ can be erroneous when using larger time steps.
		} \label{fig2:coupled-TS}
	\end{figure}
	In this test, we perform numerical simulations to study the sensitivity of numerical results to   stabilizing parameter $S$ in the schemes. We use computational domain $\Omega=[0, 1] \times [0, 1]$, three particles represented by $\phi_i, i=1,2,3,$ and the following initial conditions:
	\begin{align}
		\begin{cases}
			\bu(x, y, 0) = \b0,\\
			\phi_{ i }(x, y, 0) = 0.5\left( 1 + \tanh \frac{r_i - \sqrt{(x - x_i)^2 + (y - y_i)^2}}{\sqrt{2} \epsilon }\right),\quad i = 1, 2, 3,
		\end{cases}
	\end{align}
	where $(r_1, x_1, y_1) = (0.08, 0.2, 0.8)$, $(r_2, x_2, y_2) = (0.08, 0.5, 0.5)$, $(r_3, x_3, y_3) = (0.08, 0.8, 0.2)$ and $\epsilon = 0.01$. We set $M = 10^{-3}$, $S = \epsilon$, $\Lambda = 10$. The other model parameter values read as follows
	\begin{align}
		\begin{cases}
			\rho = 0.01,\;  \eta = 1,\;  \alpha_{1 } = 6 \pi \eta r_1 ,\;
			\alpha_{2} = 6 \pi \eta r_2,\; \alpha_{3} = 6 \pi \eta r_3,\\
			\varepsilon = 0.001,\; \gamma = 1,\; \gamma_1 = \gamma \epsilon,\; \gamma_2 = \frac{\gamma}{4\epsilon},\;   R_m = 0.4,\;  \epsilon_2 = 0.001.
		\end{cases}
	\end{align}

	\begin{figure}[H]
		\centering
		\subfigure[Energy curve with $S = \epsilon$.]{
			\includegraphics[width=0.40\textwidth,height=0.25\textwidth]{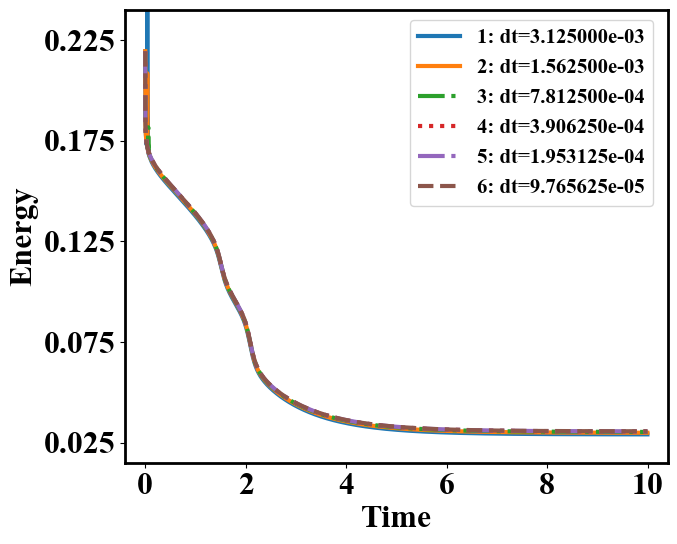}
			\includegraphics[width=0.40\textwidth,height=0.25\textwidth]{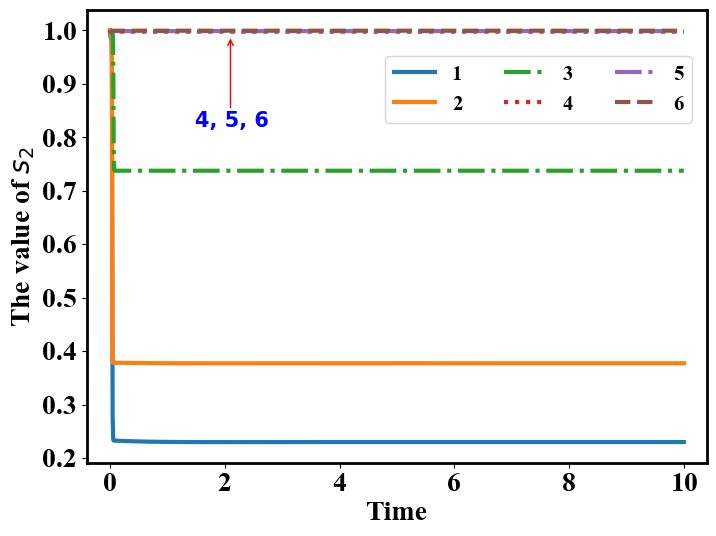}
		}
		\subfigure[Energy curve with $S = 0$.]{
			\includegraphics[width=0.40\textwidth,height=0.25\textwidth]{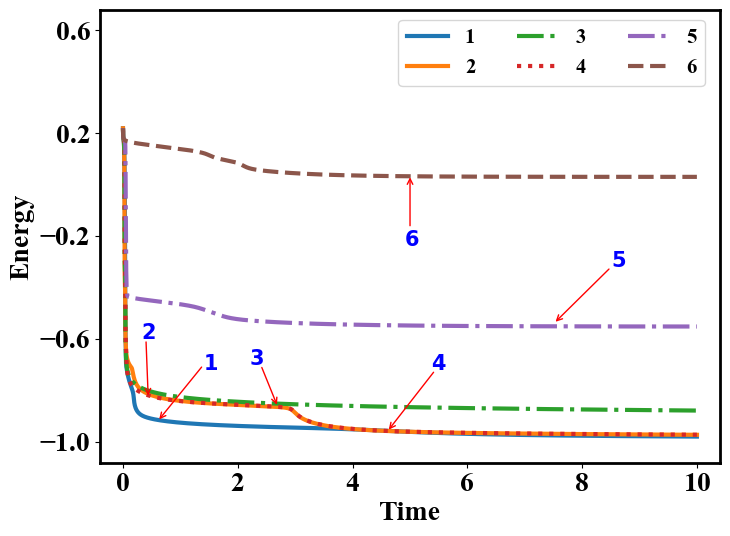}
			\includegraphics[width=0.40\textwidth,height=0.25\textwidth]{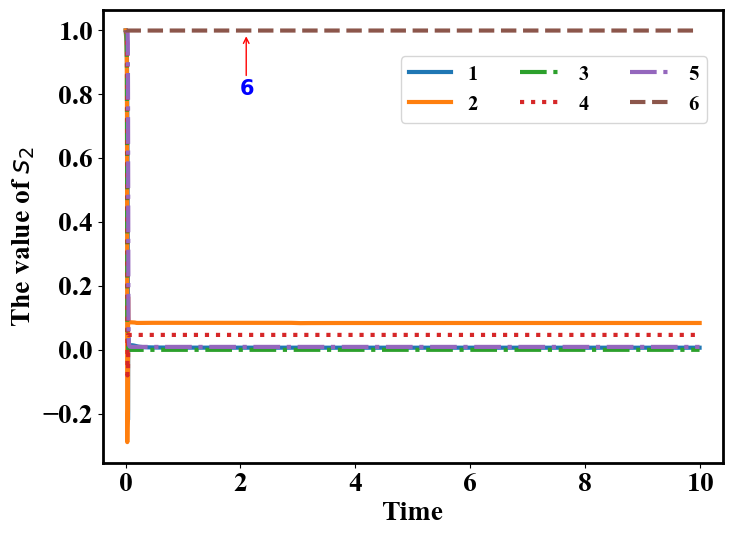}
		}
		\caption{{\bf Example \ref{eg:ST}:} Comparison of the modified energy \eqref{DSAV-energy-eq} without the  $\fr (s_2^n)^2$ term computed using the decoupled scheme for various time steps.
		} \label{fig3:decoupled-TS}
	\end{figure}
	
	\begin{figure}[H]
		\centering
		\includegraphics[width=0.46\textwidth,height=0.25\textwidth]{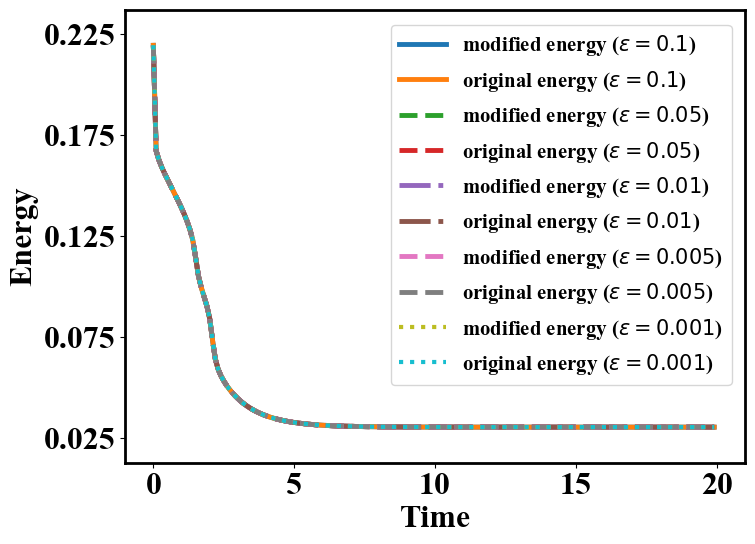}
		\caption{{\bf Example \ref{eg:ST}:} Time evolution of the original energy and modified energy  for different  $\varepsilon=0.1, 0.05, 0.01, 0.005, 0.001$. They are indistinguishable.
		} \label{fig3:TS-relax}
	\end{figure}
	
	Here, we  adopt the physical boundary condition and discretize the equations in space using $128 \times 128$ grids.  To show sensitivity of the energy calculated with respect to different $S$ and  make a comparison between the coupled and decoupled scheme, we plot curves of the total energy computed using $S = 0$ and $S =  \epsilon$ in different time steps in  Figure \ref{fig2:coupled-TS} and Figure \ref{fig3:decoupled-TS}, respectively. We summarize the results as follows.
	\begin{enumerate}
		\item
		Both the coupled and decoupled scheme with $S = 0$ can yield incorrect energy profiles for $\tau > 9.765625 \times 10^{-5}$.
		\item
		All the energy curves calculated using the coupled scheme with $S = \epsilon$ and properly selected time steps demonstrate monotonic decay, demonstrating that stabilizing methods indeed ``stabilize" the schemes at relatively large time steps. The decoupled scheme with $S = \epsilon$ is numerically ``stable" until $\tau \approx 3.125 \times 10^{-3}$.
		Although the energy curve of the decoupled scheme with $S = \epsilon$ for time step $\tau = 1.5625 \times 10^{-3}$, $\tau = 7.8125 \times 10^{-4}$ decays in time, the value of $s_2$ in each case is far away from $1$, indicating numerical errors in these results. The magnitude of $s_2$ is close to $1$, the exact solution, for $S = \varepsilon$ only when  $\tau < 3.90625  \times 10^{-4}$.  So, we conclude that the decoupling strategy makes the scheme less ``stable" and less accurate in comparison with the fully coupled scheme.
	\end{enumerate}
	This example demonstrates  that although both schemes are unconditionally energy stable theoretically and the decoupled one is easier to implement than the coupled one in practice, both schemes can produce inaccurate numerical results at relative large time steps; in comparison, the fully coupled scheme seems to be more ``stable" and accurate than the decoupled one in this test.
	
	In the following, we test the sensitivity of the stabilized scheme to parameter $S$. In order to obtain good accuracy, we take time step $\tau  = 10^{-4}$ to
	perform the following simulations using $S=\varepsilon=0.1, 0.05, 0.01, 0.005, 0.001$, respectively.
	In Figure \ref{fig3:TS-relax},  we compare evolution of the original discrete energy defined without energy quadratization and the modified discrete energy defined by \eqref{DSAV-energy-eq} without the $s_2^2/2$ term for $\varepsilon=0.1, 0.05, 0.01, 0.005, 0.001$, respectively. In all cases, the curves agree very well and show monotonic decay to the steady state in long time. This shows that the value of the scheme is not very sensitive to the magnitude of stabilizing parameter $S$ so long as it is positive.

	In the following, we also test the sensitivity with respect to different values of elastic relaxation parameter $\varepsilon$. In order to obtain accurate results, we take time step $\tau  = 10^{-4}$  to
	conduct the following simulations with a fixed stabilizing parameter $S = \epsilon$. In Figure \ref{fig3:TS-relax},  we compare time evolution of the original discrete energy defined without energy quadratization and the modified discrete energy defined by \eqref{DSAV-energy-eq} without the $s_2^2/2$ term for $\varepsilon=0.1, 0.05, 0.01, 0.005, 0.001$, respectively. In all cases, the energy results agree very well and exhibit monotonic decay to the steady state in long time. This shows that the proposed scheme is not very sensitive to elastic  parameter $\varepsilon$.
\end{example}
\begin{figure}[H]
	\centering
	\subfigure[At  $t=0, 2, 3, 4, 5$ (from left to right).]{
		\begin{minipage}[]{0.8\linewidth}
			\includegraphics[width=0.18\textwidth,height=0.18\textwidth]{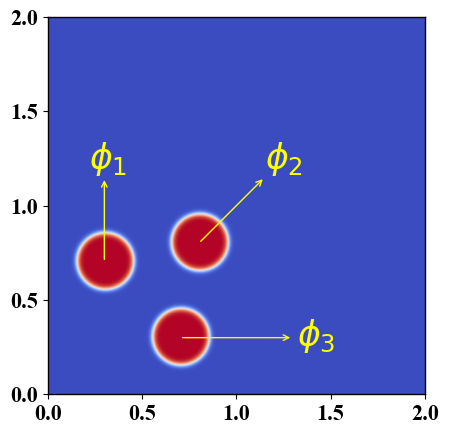}
			\includegraphics[width=0.18\textwidth,height=0.18\textwidth]{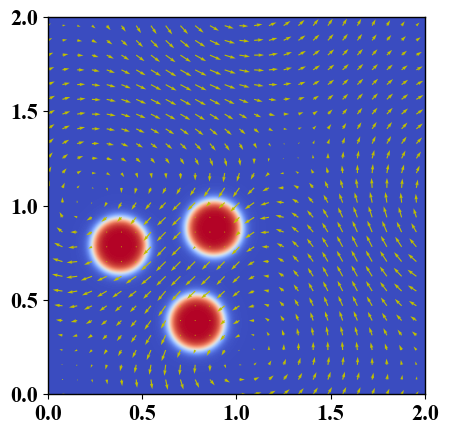}
			\includegraphics[width=0.18\textwidth,height=0.18\textwidth]{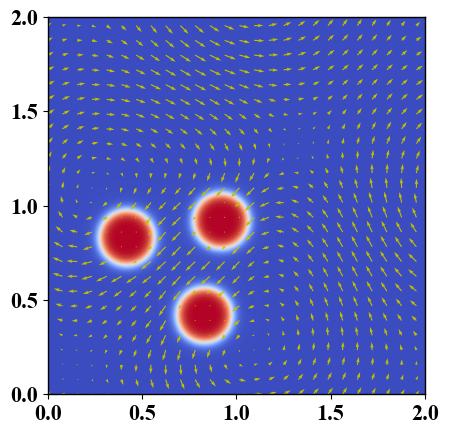}
			\includegraphics[width=0.18\textwidth,height=0.18\textwidth]{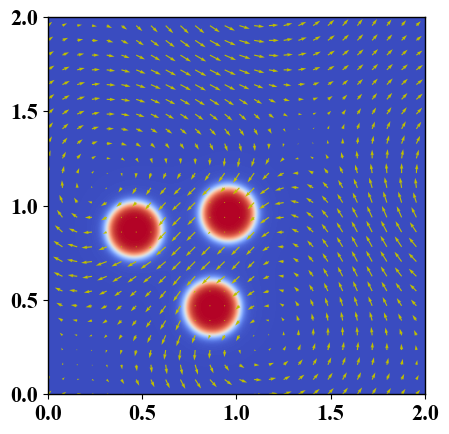}
			\includegraphics[width=0.18\textwidth,height=0.18\textwidth]{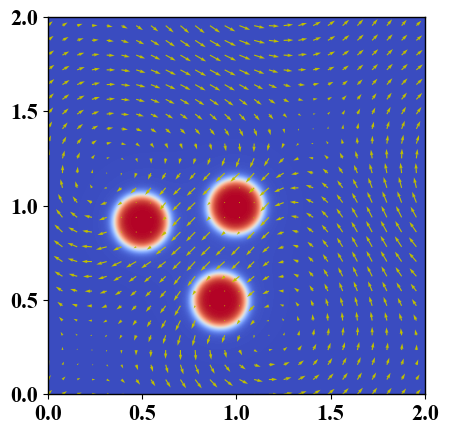}
		\end{minipage}
	}
	\subfigure[At  $t=6, 7, 8, 9, 10$ (from left to right).]{
		\begin{minipage}[]{0.8\linewidth}
			\includegraphics[width=0.18\textwidth,height=0.18\textwidth]{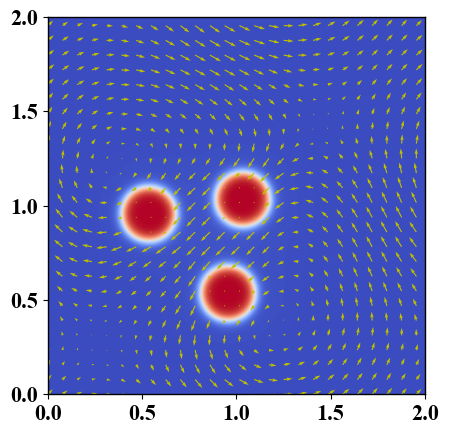}
			\includegraphics[width=0.18\textwidth,height=0.18\textwidth]{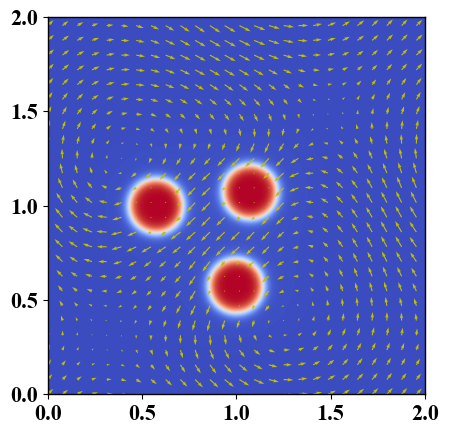}
			\includegraphics[width=0.18\textwidth,height=0.18\textwidth]{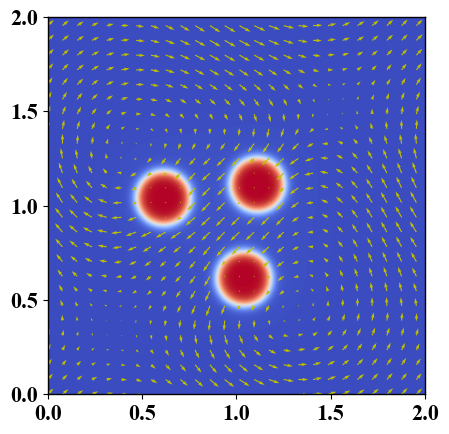}
			\includegraphics[width=0.18\textwidth,height=0.18\textwidth]{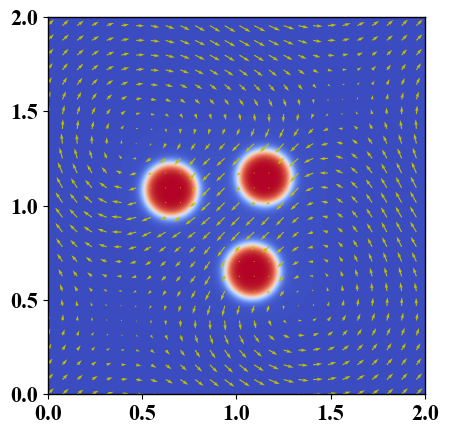}
			\includegraphics[width=0.18\textwidth,height=0.18\textwidth]{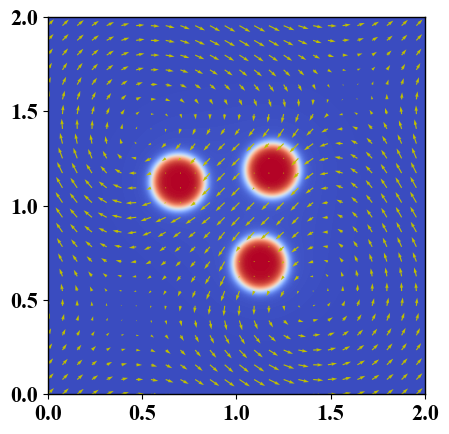}
		\end{minipage}
	}
	\subfigure[At  $t=11, 12, 13, 14, 15$ (from left to right).]{
		\begin{minipage}[]{0.8\linewidth}
			\includegraphics[width=0.18\textwidth,height=0.18\textwidth]{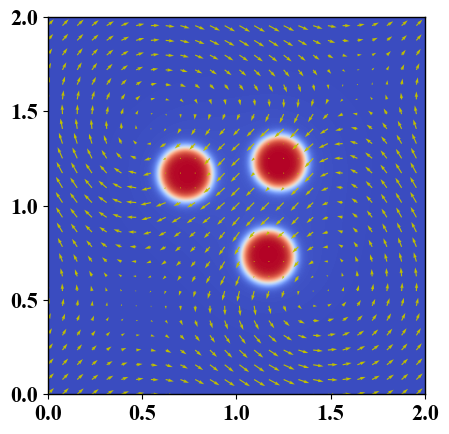}
			\includegraphics[width=0.18\textwidth,height=0.18\textwidth]{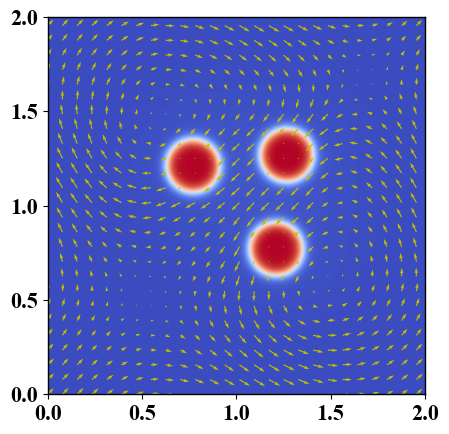}
			\includegraphics[width=0.18\textwidth,height=0.18\textwidth]{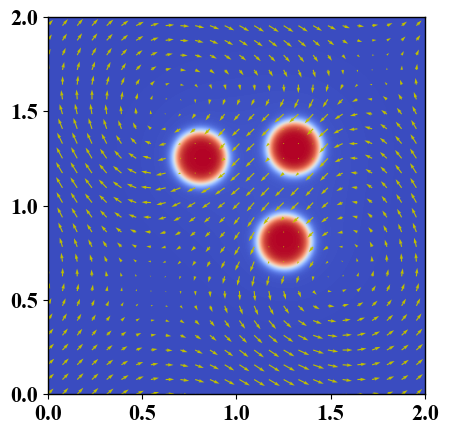}
			\includegraphics[width=0.18\textwidth,height=0.18\textwidth]{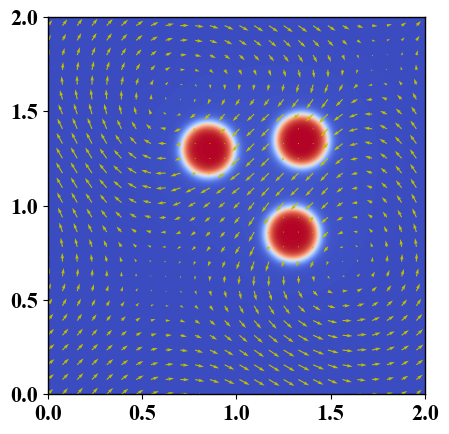}
			\includegraphics[width=0.18\textwidth,height=0.18\textwidth]{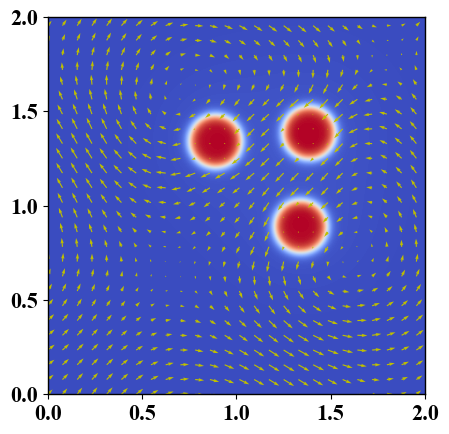}
		\end{minipage}
	}	
	\subfigure[At  $t=16, 17, 18, 19, 20$ (from left to right).]{
		\begin{minipage}[]{0.8\linewidth}
			\includegraphics[width=0.18\textwidth,height=0.18\textwidth]{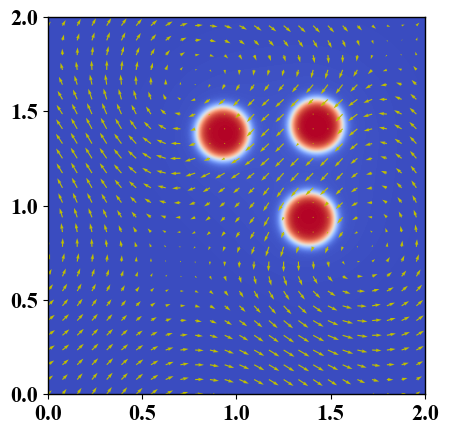}
			\includegraphics[width=0.18\textwidth,height=0.18\textwidth]{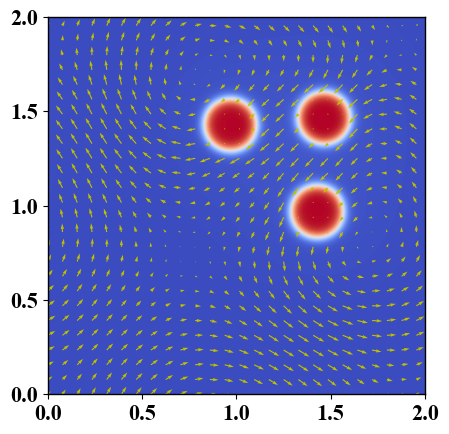}
			\includegraphics[width=0.18\textwidth,height=0.18\textwidth]{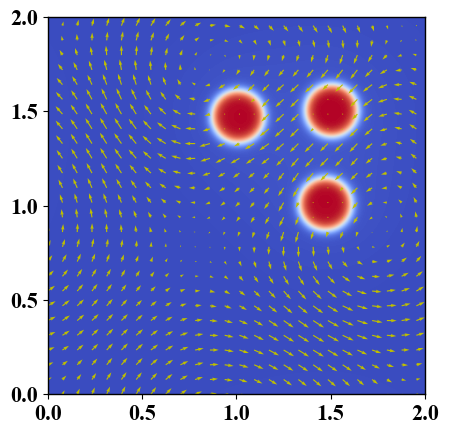}
			\includegraphics[width=0.18\textwidth,height=0.18\textwidth]{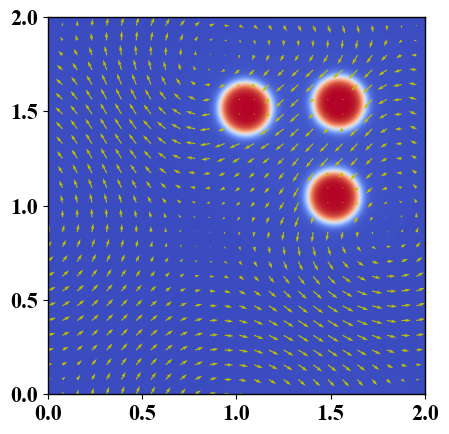}
			\includegraphics[width=0.18\textwidth,height=0.18\textwidth]{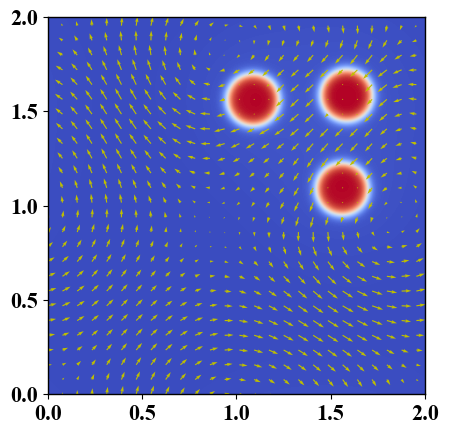}
		\end{minipage}
	}
	\caption{{\bf Example \ref{eg:3-active-particles} (Case I):} 2D dynamics of three active particles in run motion only, where several snapshots are taken from $t = 0$ to $20$. The mobility parameter is set at $M = 10^{-4}$ and the time step $\tau = 10^{-4}$. The arrows represent the velocity direction at grid points.
	}\label{fig:active-case1-pcolor}
\end{figure}
\begin{example}[Run and tumble motion of three active particles ]\label{eg:3-active-particles}
	In this example, we simulate run-and-tumble motion  of active particles, where they move at constant self-propelled speed for awhile and then change directions randomly in  time.  The kinematic velocity of the $i$-th active particle's center of mass  is defined  in \eqref{active-speed}.
	We use 2D computational domain $\Omega = [0, 2] \times [0, 2]$ with periodic
	boundary conditions and  $256 \times 256$ equal distanced meshes in space.
	The initial velocity of the particles is set to zero.
	The initial values of the phase field functions  are given by
	\begin{align}
		\phi_i(x, y, t)|_{t = 0} = 0.5 \left( 1 + \tanh \frac{r_i - \sqrt{(x - x_i)^2 + (y - y_i)^2}}{\sqrt{2} \epsilon} \right),\quad i =1, 2, 3,
	\end{align}
	where $\epsilon = 0.01$. We set $ M = 10^{-4}$, $S = 10 \epsilon$ and $\tau = 10^{-4}$. If not explicitly specified, the other parameter values are given as follows
	\begin{align}
		\begin{cases}
			\rho = 0.01,\;  \eta = 1,\;  \alpha_{1 } = 6 \pi r_1 \eta,\;
			\alpha_{2} = 6 \pi r_2 \eta,\; \alpha_{3} = 6 \pi r_3 \eta,\; \Lambda = 10,\\
			\varepsilon = 10^{-3},\; \gamma = 1,\; \gamma_1 = \gamma \epsilon,\; \gamma_2 = \frac{\gamma}{4\epsilon},\;    R_m = 0.45,\;  \epsilon_2 = 0.001.
		\end{cases}
	\end{align}

	\noindent {\bf Case I: Run motion.} Firstly, we consider the run motion of three active particles  with the same velocity initially.
	In this case, we set $ (r_1, x_1, y_1) = (0.15, 0.3, 0.7)$, $(r_2, x_2, y_2) = (0.15, 0.8, 0.8)$, $(r_3, x_3, y_3)$ = $(0.15, 0.7, 0.3)$ and the active velocity is given by $\bp_i = (0.05, 0.05)$, $i=1, 2, 3$. The velocity and vorticity field are shown in Figure \ref{fig:active-case1-pcolor} and Figure \ref{fig:active-case1-vorticity}, respectively. We observe that the three active particles are initially located at the southwest corner, then they self-propel themselves toward the northeast. During the process, we observe that the moving direction of the particles is opposite to that of the flow field around them, and each particle maintains their shapes, which confirms that rigidity constraint \eqref{D-eq}  is very well maintained. In addition,  we also see that
	a weak fluid flow is induced in the four areas that located in the front, back, left and right of the moving direction of the particles. In Figure \ref{fig:invariant-case1}, we plot  evolution of the modified energy and the volume of the three particles. We observe that the energy   decays monotonically during this process. Figure \ref{fig:invariant-case1} (b) shows that the  total volume is preserved extremely well, implying that the each material phase  is conserved accurately.    Figure \ref{fig:invariant-case1} (c) demonstrates that $s_2$ is always close to $1$ in the simulation, confirming  the accuracy of the simulation.
	\begin{figure}[H]
		\centering
		\subfigure[Snapshots of vorticity at $t=0.1, 2, 3, 4, 5$.]{
			\begin{minipage}[]{0.8\linewidth}	
				\includegraphics[width=0.17\textwidth,height=0.146\textwidth]{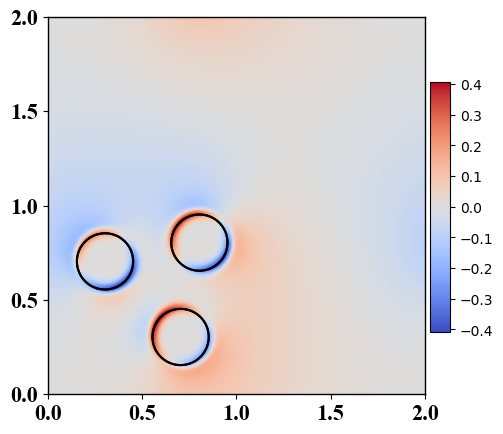}
				\includegraphics[width=0.17\textwidth,height=0.146\textwidth]{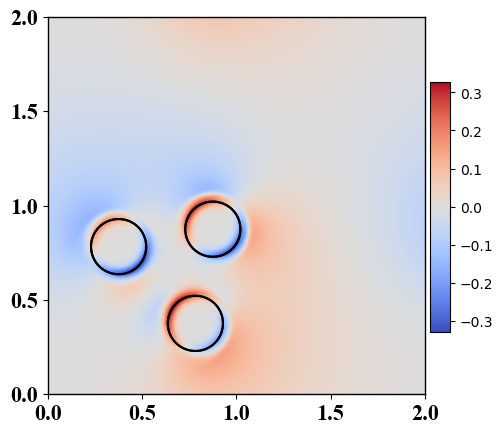}
				\includegraphics[width=0.17\textwidth,height=0.146\textwidth]{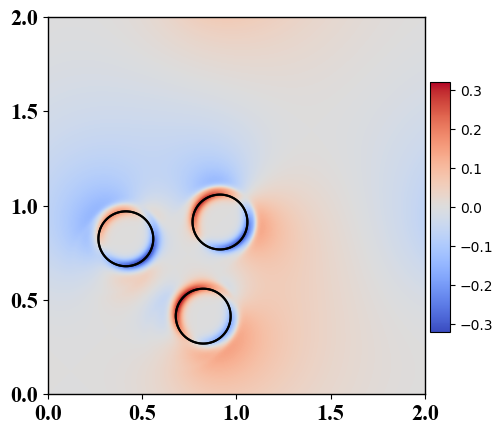}
				\includegraphics[width=0.17\textwidth,height=0.146\textwidth]{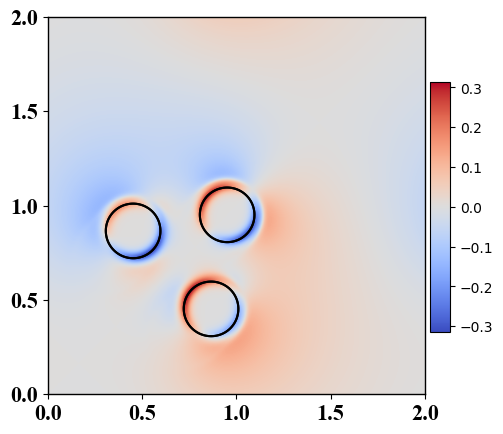}
				\includegraphics[width=0.17\textwidth,height=0.146\textwidth]{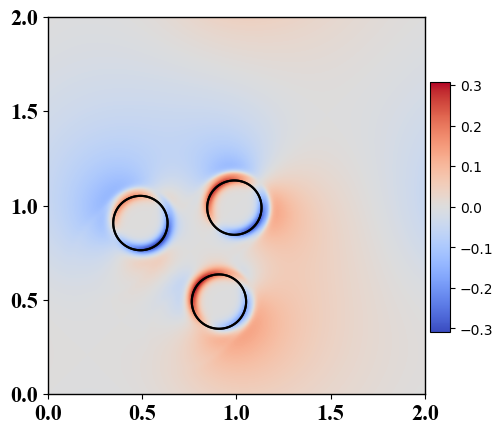}
			\end{minipage}
		}
		\subfigure[Snapshots of vorticity at $t=6, 7, 8, 9, 10$.]{
			\begin{minipage}[]{0.8\linewidth}	
				\includegraphics[width=0.17\textwidth,height=0.146\textwidth]{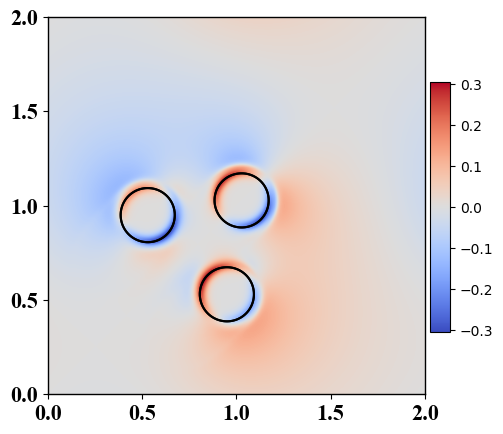}
				\includegraphics[width=0.17\textwidth,height=0.146\textwidth]{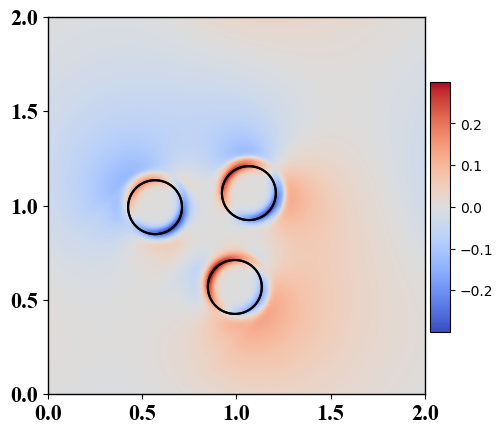}
				\includegraphics[width=0.17\textwidth,height=0.146\textwidth]{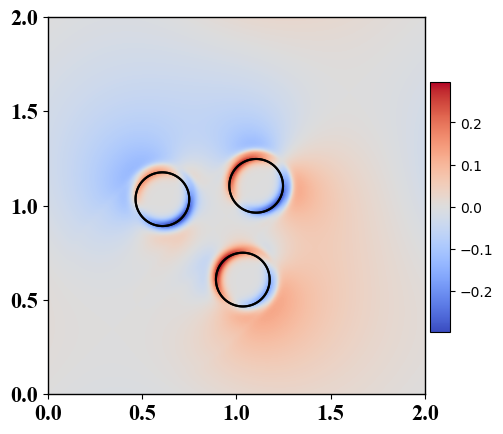}
				\includegraphics[width=0.17\textwidth,height=0.146\textwidth]{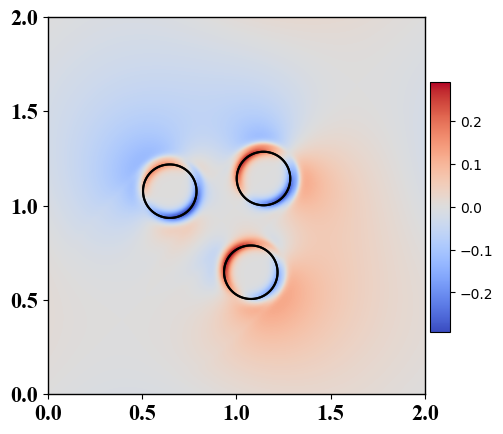}
				\includegraphics[width=0.17\textwidth,height=0.146\textwidth]{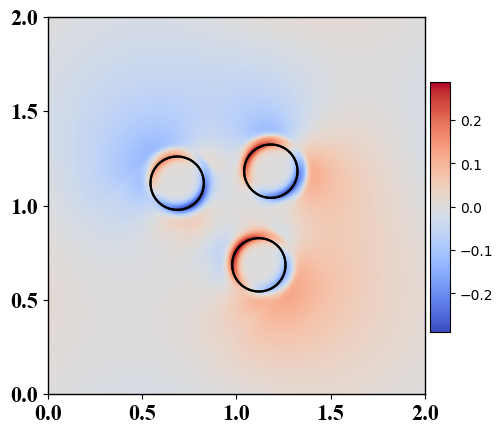}
			\end{minipage}	
		}
		\subfigure[Snapshots of vorticity at $t=11, 12, 13, 14, 15$.]{
			\begin{minipage}[]{0.8\linewidth}		
				\includegraphics[width=0.17\textwidth,height=0.146\textwidth]{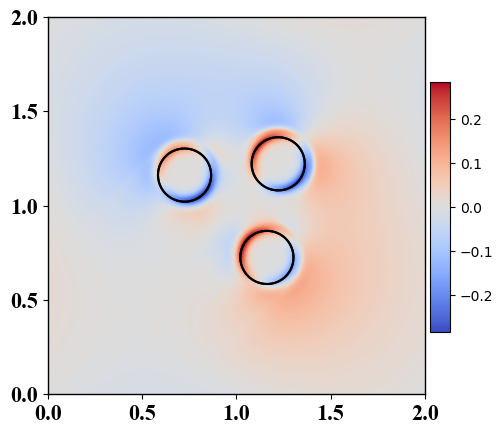}
				\includegraphics[width=0.17\textwidth,height=0.146\textwidth]{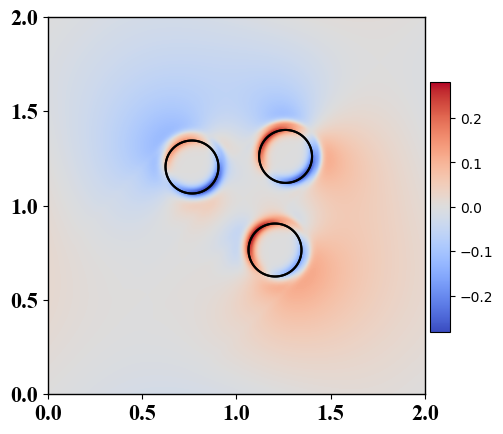}
				\includegraphics[width=0.17\textwidth,height=0.146\textwidth]{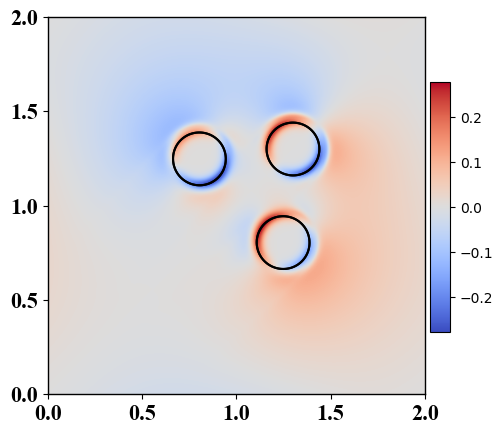}
				\includegraphics[width=0.17\textwidth,height=0.146\textwidth]{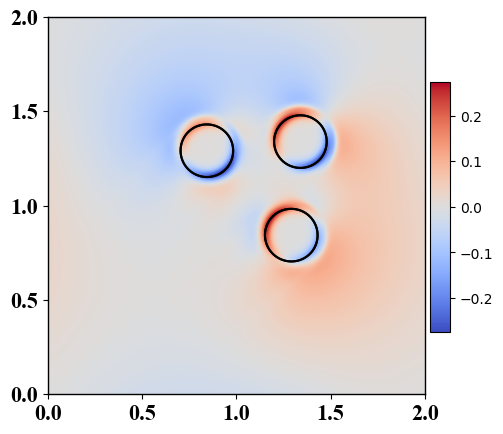}
				\includegraphics[width=0.17\textwidth,height=0.146\textwidth]{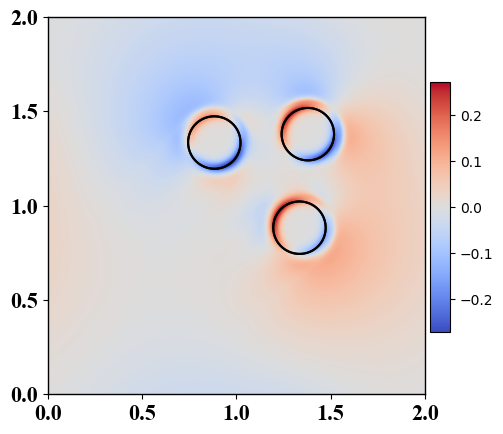}
			\end{minipage}
		}
		\subfigure[Snapshots of vorticity at $t=16, 17, 18, 19, 20$.]{
			\begin{minipage}[]{0.8\linewidth}	
				\includegraphics[width=0.17\textwidth,height=0.146\textwidth]{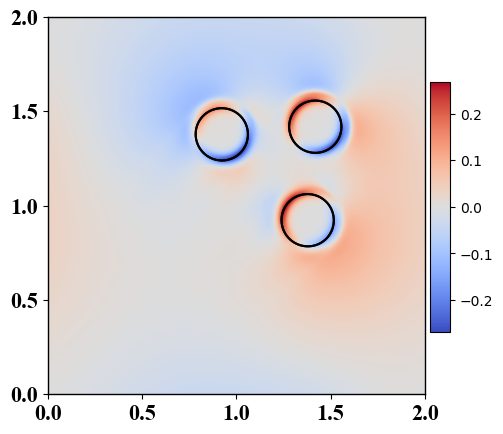}
				\includegraphics[width=0.17\textwidth,height=0.146\textwidth]{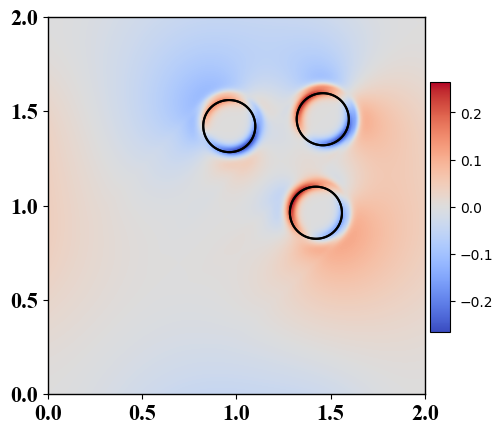}
				\includegraphics[width=0.17\textwidth,height=0.146\textwidth]{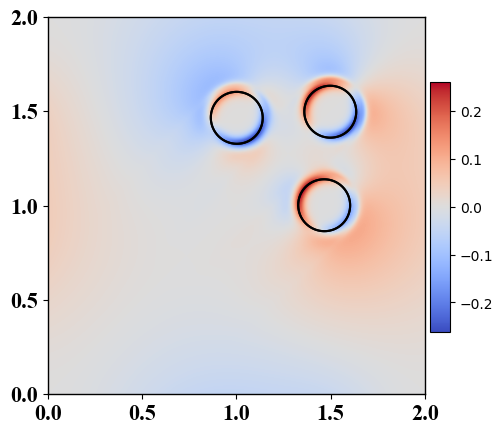}
				\includegraphics[width=0.17\textwidth,height=0.146\textwidth]{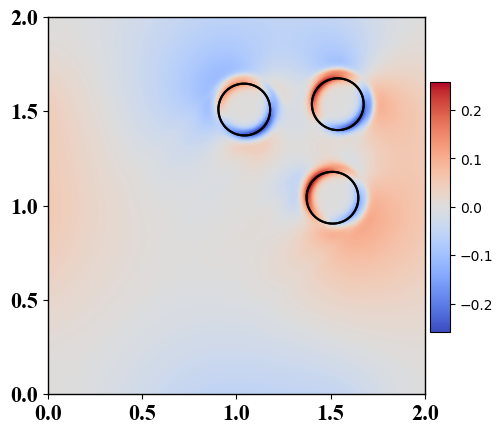}
				\includegraphics[width=0.17\textwidth,height=0.146\textwidth]{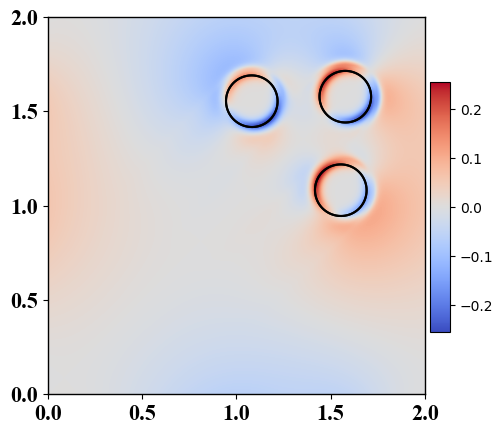}
			\end{minipage}
		}
		\caption{{\bf Example \ref{eg:3-active-particles} (Case I):} 2D snapshots of the vorticity field around the three active particles in the run motion, at selected moments between $t=0.1$ and $20$ with $M = 10^{-4}$ and $\tau = 10^{-4}$. The contour of $\phi_{ i } = 0.5$ is shown in black.} \label{fig:active-case1-vorticity}
	\end{figure}
	
	\begin{figure}[H]
		\centering
		\subfigure[Energy plot vs. time.]{
			\includegraphics[width=0.30\textwidth,height=0.25\textwidth]{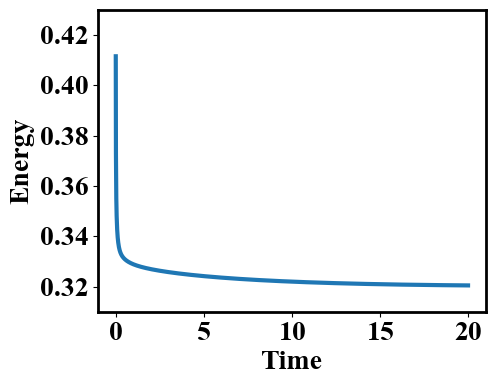}
		}
		\subfigure[Total volume change vs time.]{
			\includegraphics[width=0.30\textwidth,height=0.25\textwidth]{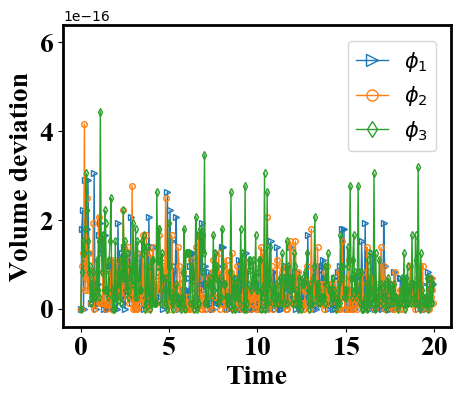}
		}
		\subfigure[ $s_2$ vs time.]{
			\includegraphics[width=0.30\textwidth,height=0.25\textwidth]{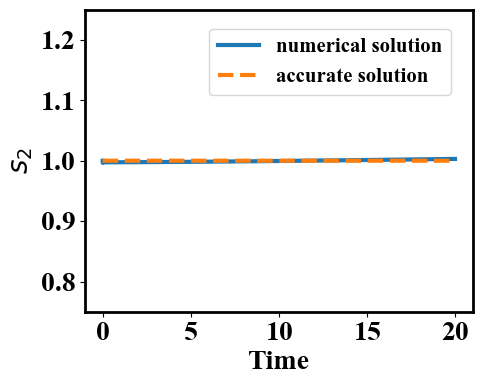}
		}
		\caption{{\bf Example \ref{eg:3-active-particles}:} $(a)$. Time evolution of the modified energy   and volume in the run motion of three active particles. $(b)$. This subfigure shows volume conservation for the three active particles. $(c)$. The numerical solution of $s_2(t)$ is very close to $1$.} \label{fig:invariant-case1}
	\end{figure}
%
%
	\begin{figure}[H]
		\centering
		\subfigure[At   $t=0, 2, 3, 4, 5$ (from left to right).]{
			\begin{minipage}[]{0.8\linewidth}	
				\includegraphics[width=0.18\textwidth,height=0.18\textwidth]{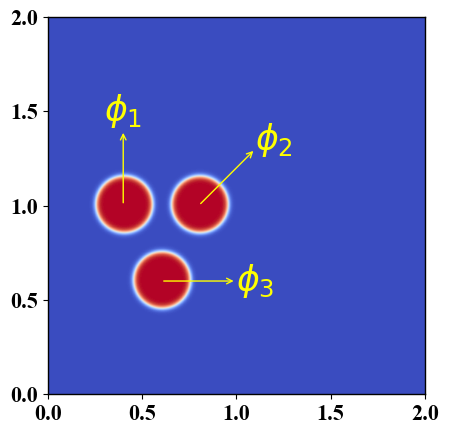}
				\includegraphics[width=0.18\textwidth,height=0.18\textwidth]{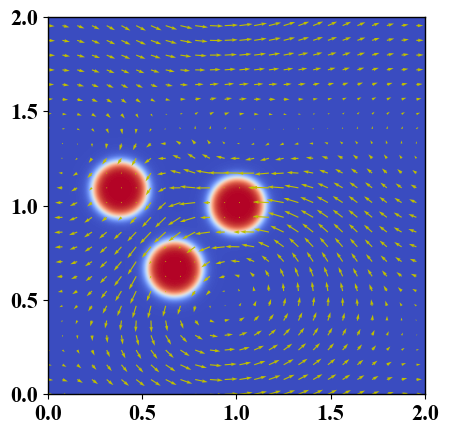}
				\includegraphics[width=0.18\textwidth,height=0.18\textwidth]{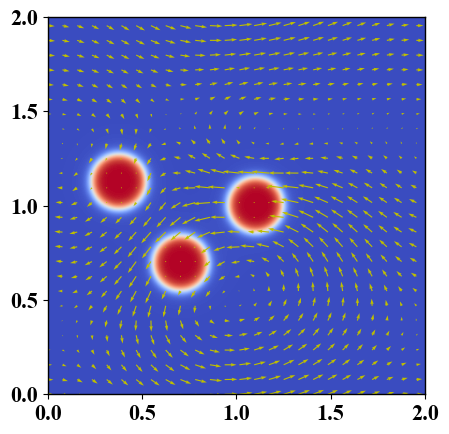}
				\includegraphics[width=0.18\textwidth,height=0.18\textwidth]{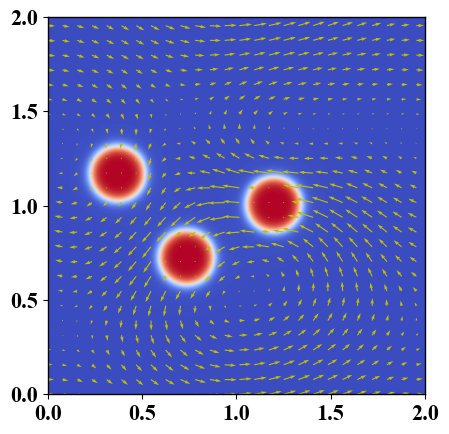}
				\includegraphics[width=0.18\textwidth,height=0.18\textwidth]{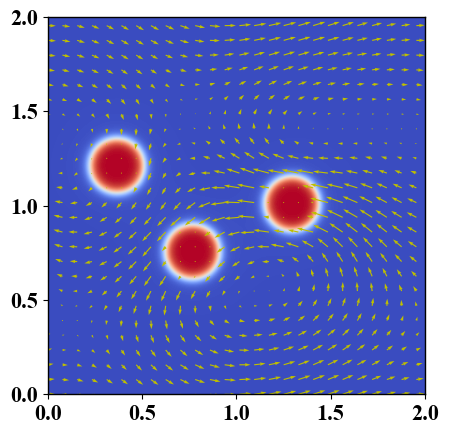}
			\end{minipage}
		}
		\subfigure[At   $t=6, 7, 8, 9, 10$ (from left to right).]{
			\begin{minipage}[]{0.8\linewidth}		
				\includegraphics[width=0.18\textwidth,height=0.18\textwidth]{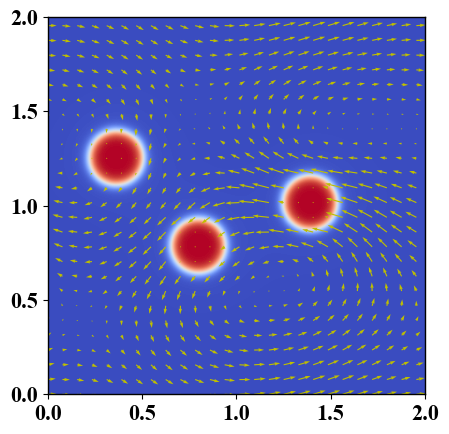}
				\includegraphics[width=0.18\textwidth,height=0.18\textwidth]{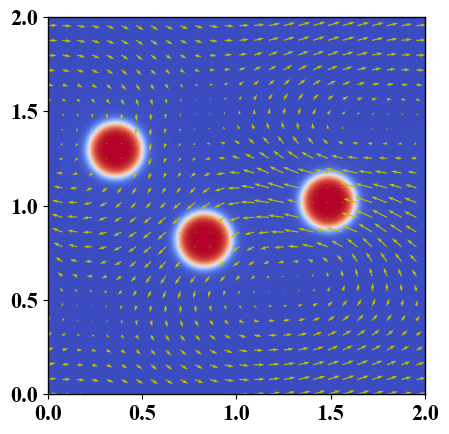}
				\includegraphics[width=0.18\textwidth,height=0.18\textwidth]{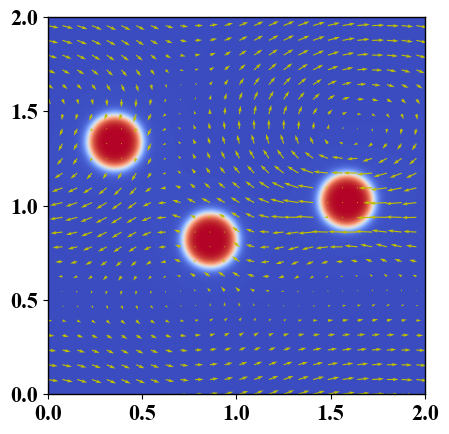}
				\includegraphics[width=0.18\textwidth,height=0.18\textwidth]{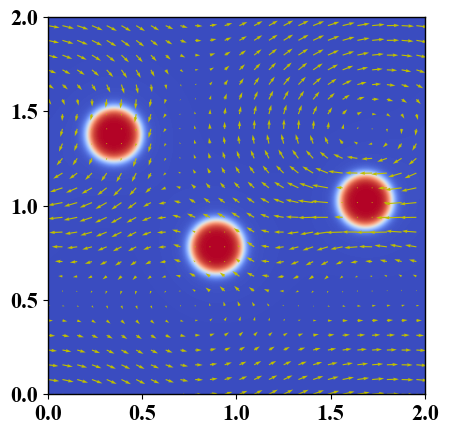}
				\includegraphics[width=0.18\textwidth,height=0.18\textwidth]{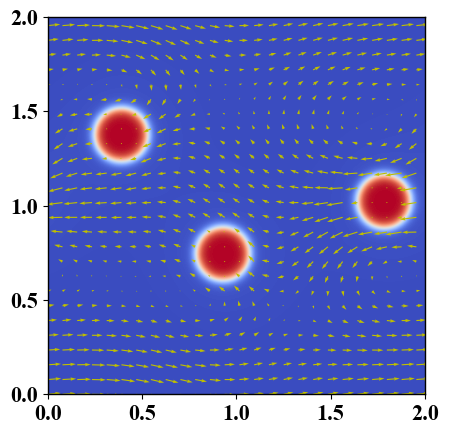}
			\end{minipage}	
		}
		\subfigure[At   $t=11, 12, 13, 14, 15$ (from left to right).]{
			\begin{minipage}[]{0.8\linewidth}		
				\includegraphics[width=0.18\textwidth,height=0.18\textwidth]{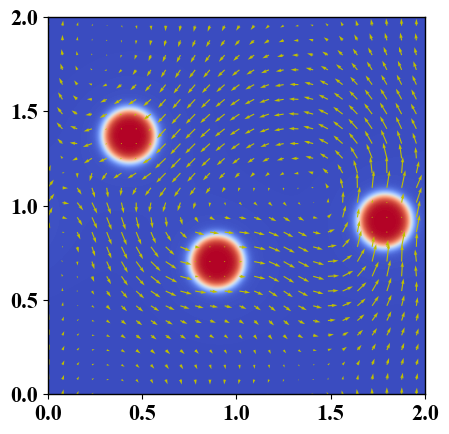}
				\includegraphics[width=0.18\textwidth,height=0.18\textwidth]{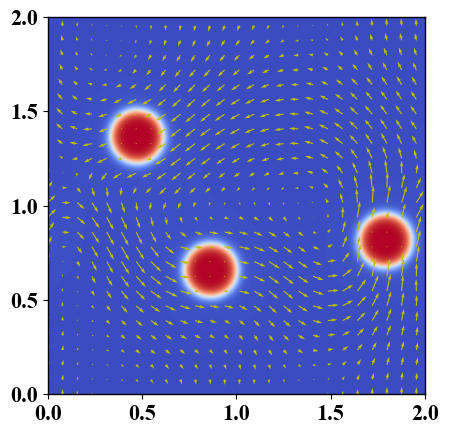}
				\includegraphics[width=0.18\textwidth,height=0.18\textwidth]{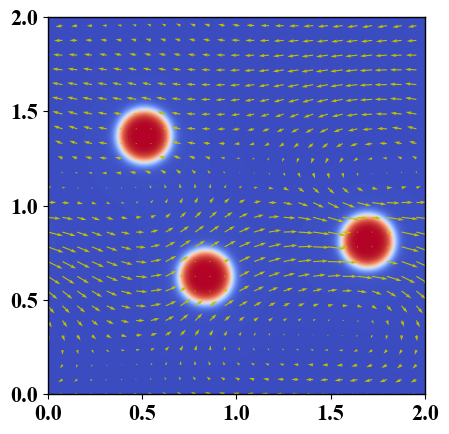}
				\includegraphics[width=0.18\textwidth,height=0.18\textwidth]{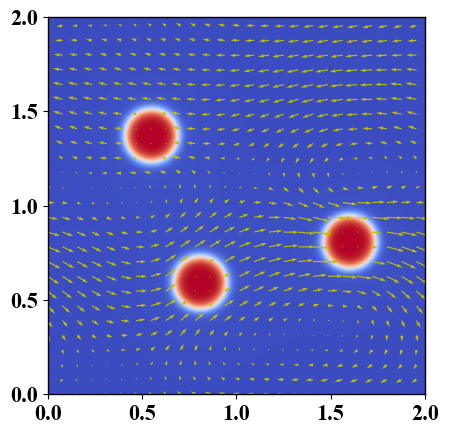}
				\includegraphics[width=0.18\textwidth,height=0.18\textwidth]{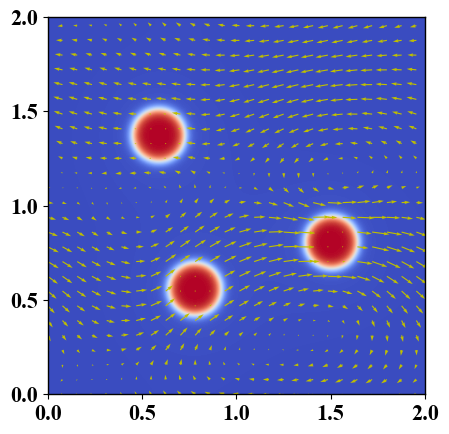}
			\end{minipage}	
		}
		\subfigure[At   $t=16, 17, 18, 19, 20$ (from left to right).]{
			\begin{minipage}[]{0.8\linewidth}		
				\includegraphics[width=0.18\textwidth,height=0.18\textwidth]{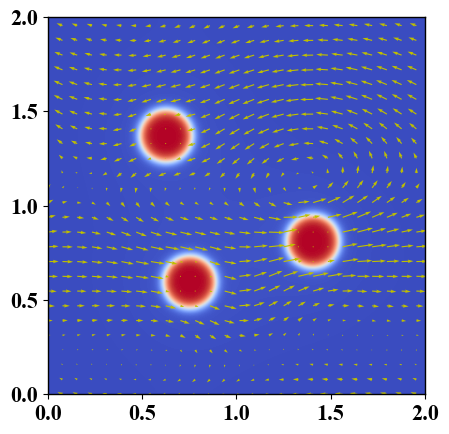}
				\includegraphics[width=0.18\textwidth,height=0.18\textwidth]{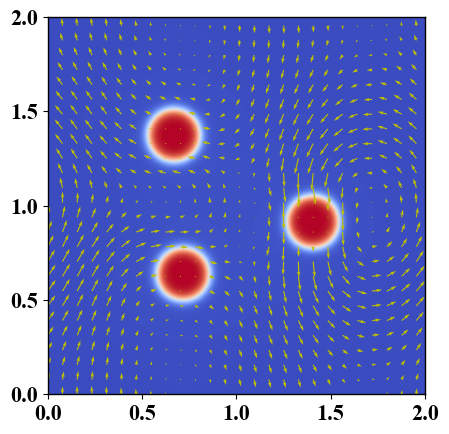}
				\includegraphics[width=0.18\textwidth,height=0.18\textwidth]{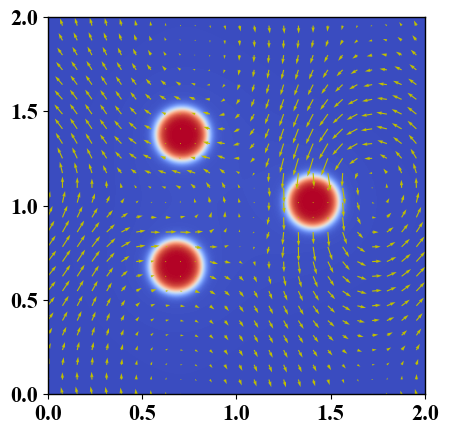}
				\includegraphics[width=0.18\textwidth,height=0.18\textwidth]{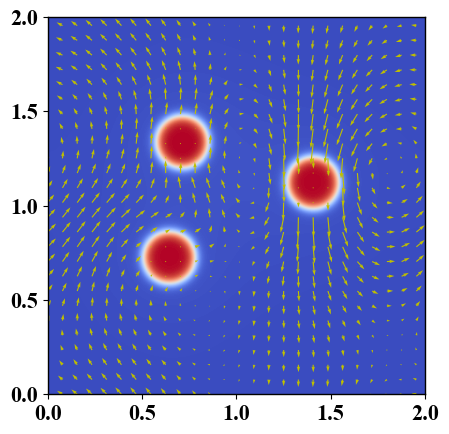}
				\includegraphics[width=0.18\textwidth,height=0.18\textwidth]{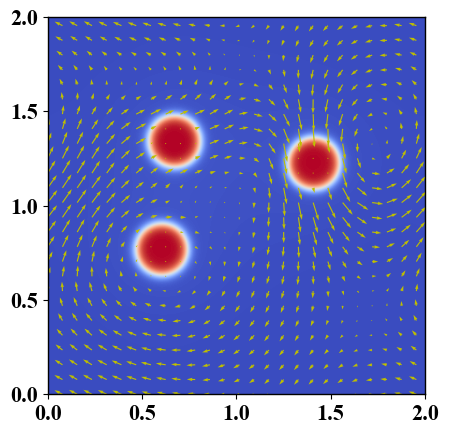}
			\end{minipage}	
		}
		\caption{{\bf Example \ref{eg:3-active-particles} (Case II):} 2D dynamical evolution of particle  profile $\phi_i$ $(i=1, 2, 3)$ for three active particles, where some snapshots are taken from $t = 0$ to $20$. The mobility parameter is $M = 10^{-4}$ and time step $\tau = 10^{-4}$, where the arrows represent the velocity director.
		} \label{fig1:active-case2-pcolor}
	\end{figure}
	\noindent {\bf Case II: run and tumble motion.} Next, we simulate the run and tumble motion of three active particles. We set $ (r_1, x_1, y_1)$ = $(0.15, 0.40, 1.00)$, $(r_2, x_2, y_2) = (0.15, 0.80, 1.00)$, $(r_3, x_3, y_3)$ = $(0.15, 0.60, 0.60)$,  denote the self-propelling velocity of the particles as $\bp_i(t)$, $i=1, 2, 3$ and sample $\bp_i(t)$ with respect to the Poisson distribution in time to determine when it tumbles and their dirctions, where the three active speeds are $\|\bp_1\| = 0.05$, $\|\bp_2\| = 0.12$ and $\|\bp_3\| = 0.06$, respectively.  In Figure \ref{fig1:active-case2-pcolor} and Figure \ref{fig1:active-case2}, we depict the velocity  and vorticity field, respectively around the three self-propelling active particles.  From the numerical experiment, we observe that the active particles run and tumble in time, inducing a changing flow field around them. Moreover, a much more heterogeneous flow field is induced in this case than in the run motion alone and  vortices are formed near the particles. 
	Once again, the particle shapes are maintained very well during the numerical simulation.
	Time evolution of the modified free energy for this case is plotted in Figure \ref{fig:invariant-tumble}, which shows the energy decays with time monotonically. Figure \ref{fig:invariant-tumble} shows that our numerical scheme preserves the total volume very well and the value of $s_2$ is maintained very close to $1$.

	\begin{figure}[H]
		\centering
		\subfigure[Snapshots of vorticity at $t=0.1, 2, 3, 4, 5$.]{
			\begin{minipage}[]{0.8\linewidth}		
				\includegraphics[width=0.17\textwidth,height=0.146\textwidth]{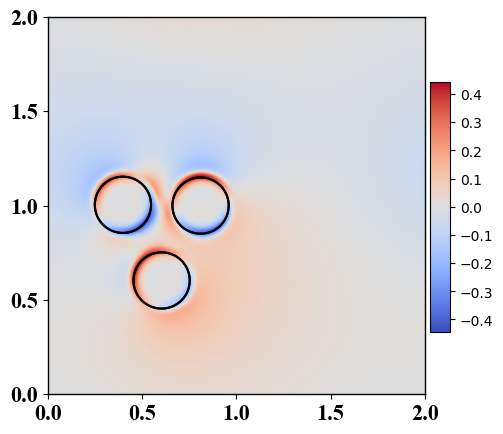}
				\includegraphics[width=0.17\textwidth,height=0.146\textwidth]{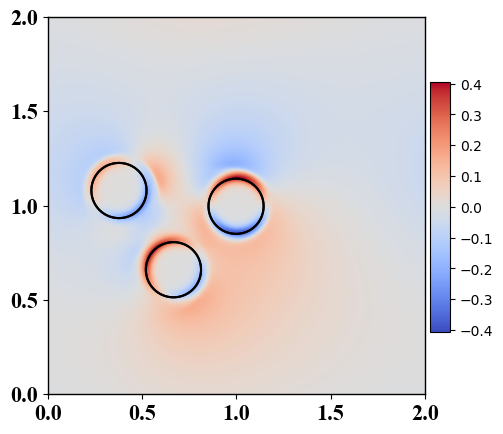}
				\includegraphics[width=0.17\textwidth,height=0.146\textwidth]{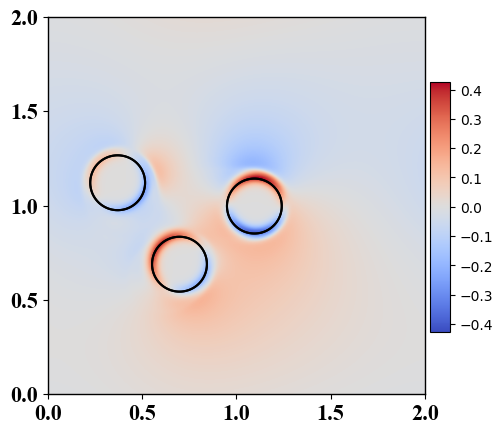}
				\includegraphics[width=0.17\textwidth,height=0.146\textwidth]{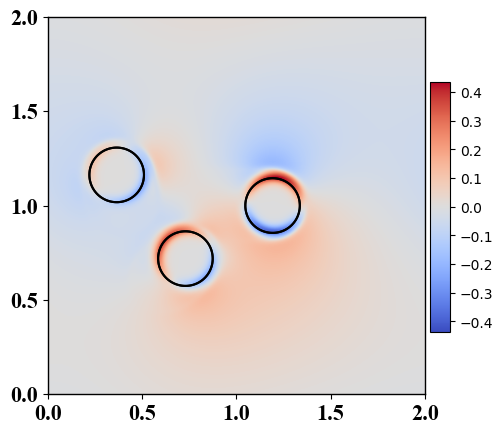}
				\includegraphics[width=0.17\textwidth,height=0.146\textwidth]{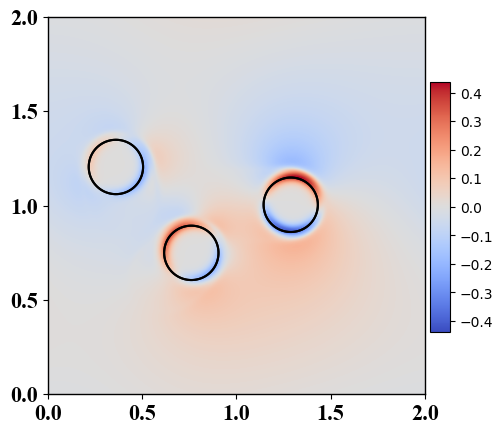}
			\end{minipage}	
		}
		\subfigure[Snapshots of vorticity at $t=6, 7, 8, 9, 10$.]{
			\begin{minipage}[]{0.8\linewidth}	
				\includegraphics[width=0.17\textwidth,height=0.146\textwidth]{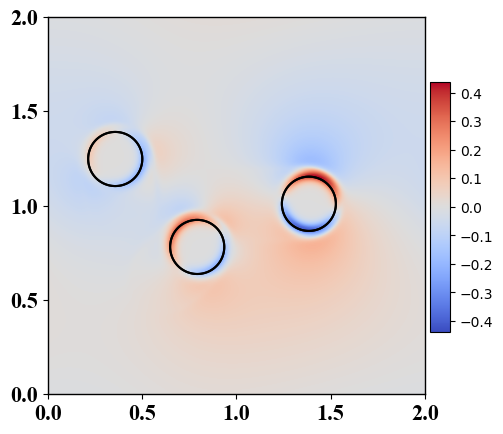}
				\includegraphics[width=0.17\textwidth,height=0.146\textwidth]{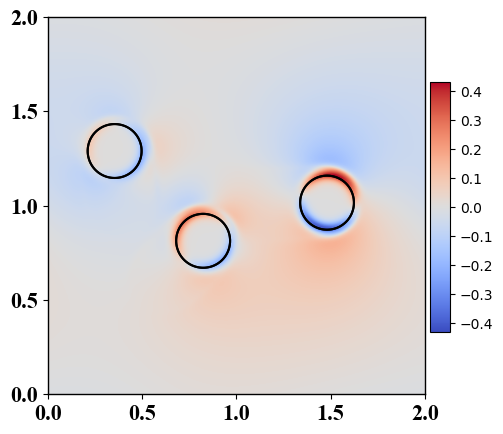}
				\includegraphics[width=0.17\textwidth,height=0.146\textwidth]{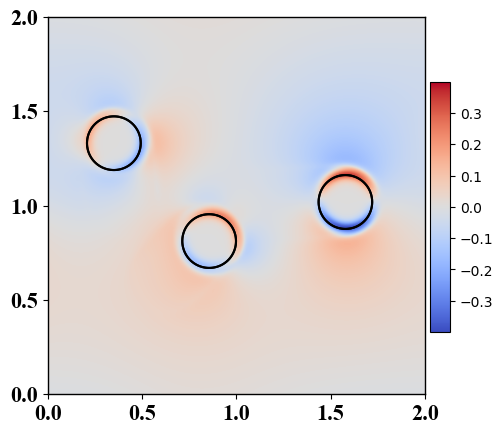}
				\includegraphics[width=0.17\textwidth,height=0.146\textwidth]{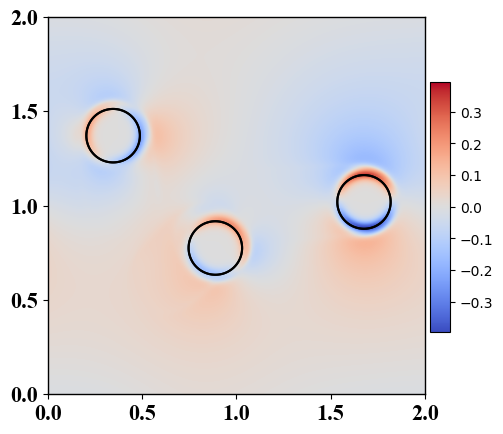}
				\includegraphics[width=0.17\textwidth,height=0.146\textwidth]{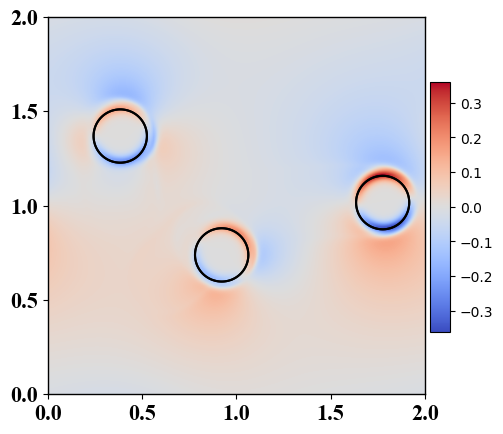}
			\end{minipage}
		}
		\subfigure[Snapshots of vorticity at $t=11, 12, 13, 14, 15$.]{
			\begin{minipage}[]{0.8\linewidth}		
				\includegraphics[width=0.17\textwidth,height=0.146\textwidth]{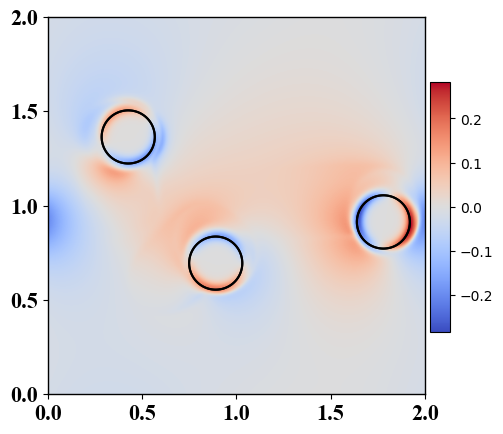}
				\includegraphics[width=0.17\textwidth,height=0.146\textwidth]{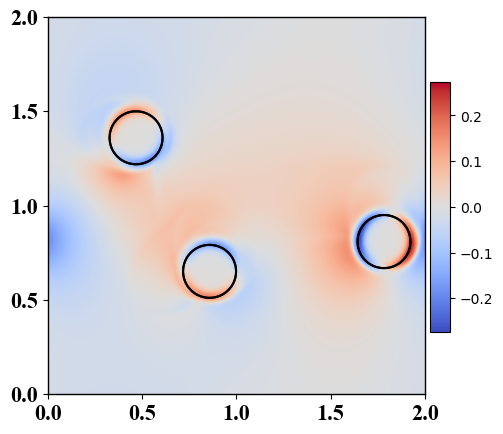}
				\includegraphics[width=0.17\textwidth,height=0.146\textwidth]{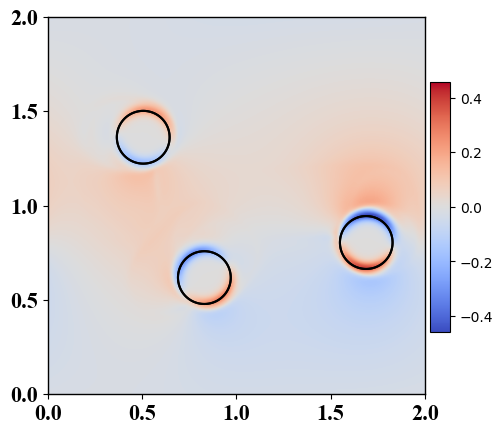}
				\includegraphics[width=0.17\textwidth,height=0.146\textwidth]{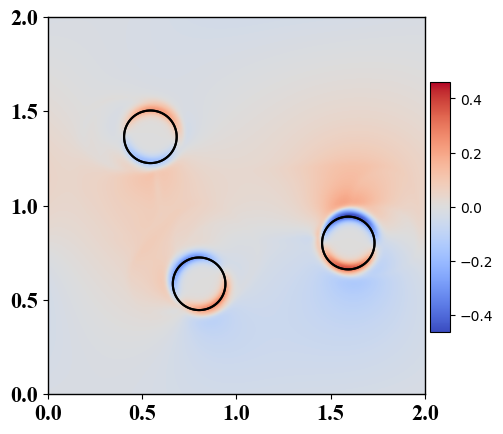}
				\includegraphics[width=0.17\textwidth,height=0.146\textwidth]{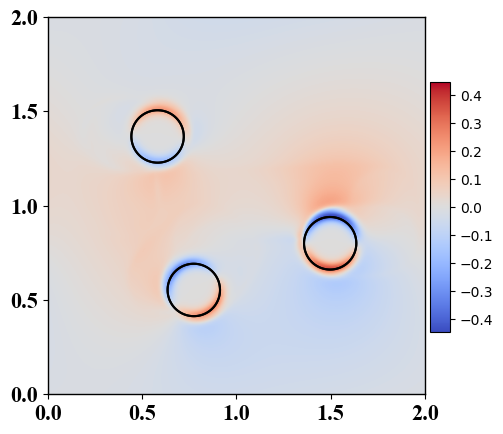}
			\end{minipage}
		}
		\subfigure[Snapshots of vorticity at $t=16, 17, 18, 19, 20$.]{
			\begin{minipage}[]{0.8\linewidth}	
				\includegraphics[width=0.17\textwidth,height=0.146\textwidth]{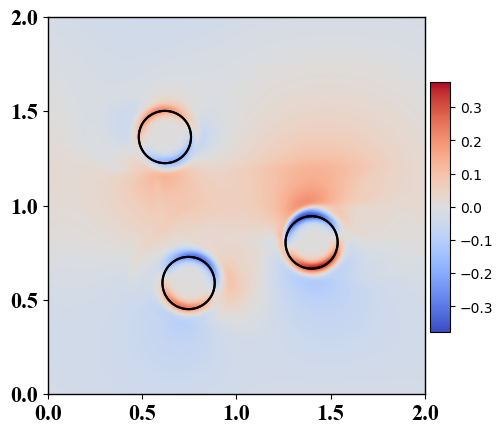}
				\includegraphics[width=0.17\textwidth,height=0.146\textwidth]{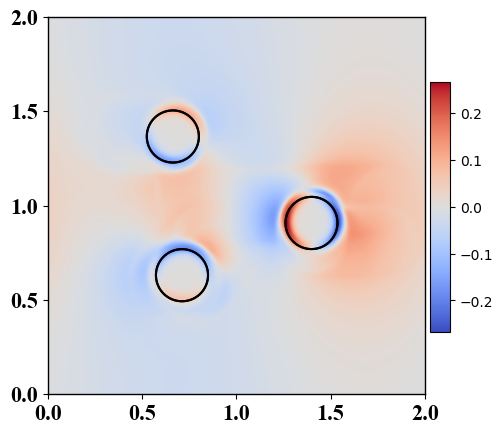}
				\includegraphics[width=0.17\textwidth,height=0.146\textwidth]{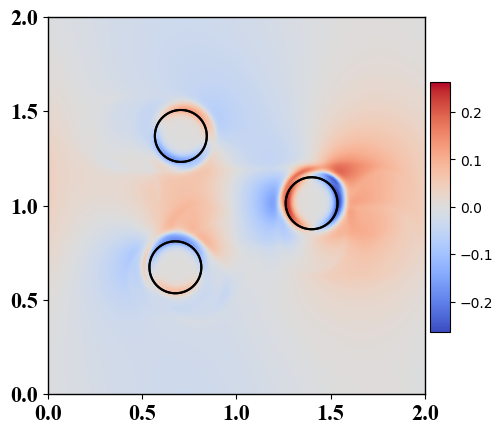}
				\includegraphics[width=0.17\textwidth,height=0.146\textwidth]{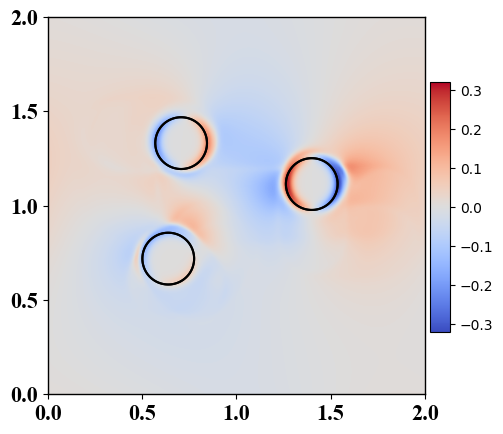}
				\includegraphics[width=0.17\textwidth,height=0.146\textwidth]{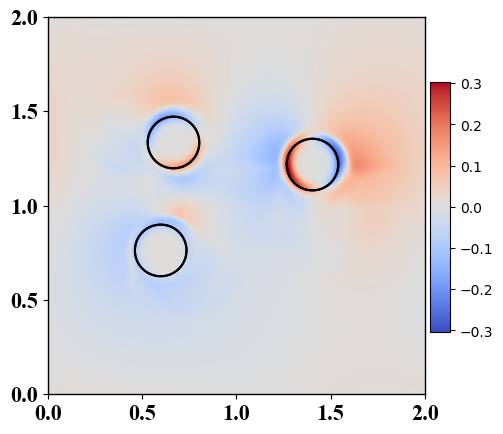}
			\end{minipage}	
		}
		\caption{{\bf Example \ref{eg:3-active-particles} (Case II):} 2D snapshots of the vorticity field around the three active particles at a sequence of time slots from $t=0.1$ to $20$ with $M = 10^{-4}$ and $\tau = 10^{-4}$. The contour of $\phi_{ i } = 0.5$ is shown in black.} \label{fig1:active-case2}
	\end{figure}
	
	\begin{figure}[H]
		\centering
		\subfigure[Energy plot vs. time.]{
			\includegraphics[width=0.30\textwidth,height=0.25\textwidth]{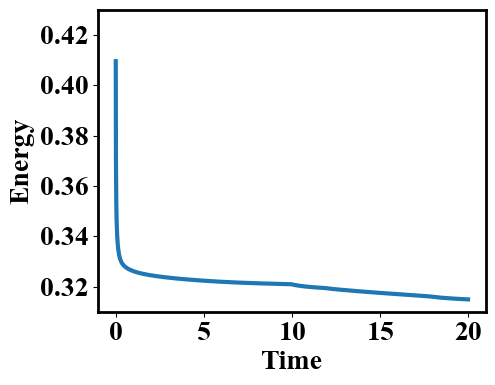}
		}
		\subfigure[Total volume change vs time.]{
			\includegraphics[width=0.30\textwidth,height=0.25\textwidth]{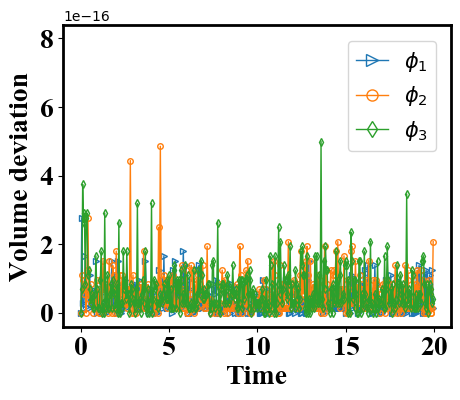}
		}
		\subfigure[$s_2$ vs time.]{
			\includegraphics[width=0.30\textwidth,height=0.25\textwidth]{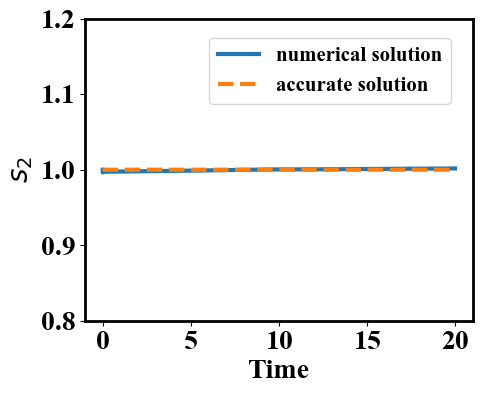}
		}	
		\caption{{\bf Example \ref{eg:3-active-particles} (\textrm{case II}):} $(a)$. Time evolution of the modified energy. $(b)$. Volume conservation for the three active particles. $(c)$. The numerical solution of $s_2(t)$ accurately approximates its exact solution.} \label{fig:invariant-tumble}
	\end{figure}	

\end{example}
%
Finally, we perform a 3D simulation to showcase dynamics of an active particle colliding with a fixed spherical obstacle. In this numerical simulation, we set the self-propelling velocity of a small active particle at $\bp_1 = (0, 0, 0.05)^T$ initially and fix a large spherical ball in the computational domain as the obstacle. After the collision, the small active particle bounces back with a randomly reflected  self-propelling velocity.   The computational domain is   $\Omega=[0, 1] \times [0, 1] \times [0, 2]$ and the initial conditions are given as follows
\begin{align}
	\begin{cases}
		\bu^0(\bx) = \b0, \\
		\phi_i^0(\bx) = 0.5 \left( 1 + \tanh \frac{r_i - \sqrt{(x - x_i)^2 + (y - y_i)^2}}{\sqrt{2} \epsilon} \right),\quad i =1, 2,
	\end{cases}
\end{align}
where $\epsilon = 0.025$, $r_1 = 0.15$, $r_2 = 0.3$, $x_1 = x_2 = y_1 = y_2 = 0.5$, $z_1 = 0.3$, $z_2 = 1.0$. We choose model parameter values as
$M = 10^{-5}$, $R_m = 0.5$, $S = 5 \epsilon$, $\epsilon_2 = 0.01$. The other parameters are the same as we used in the previous example. To conduct this 3D simulation, we use uniform meshes with spatial step size $h_x = h_y = h_z = 1/64$. The dynamical process of a particle colliding with an obstacle is depicted in Figure \ref{fig:3DFSI}, where all the plots are viewed from the same angle. We observe that the small spherical particle rises approaching  to the larger spherical obstacle due to self-propulsion, collision occurs at about $t = 5$, and then active particle bounces back. Since the reflective velocity is chosen randomly in the simulation, we notice that the position of the active particle at $t =10$ is closer to the right in stead at the middle.
\begin{figure}[H]
	\centering
	\subfigure[$t = 0$.]{
		\includegraphics[width=0.18\textwidth, height=0.29\textwidth]{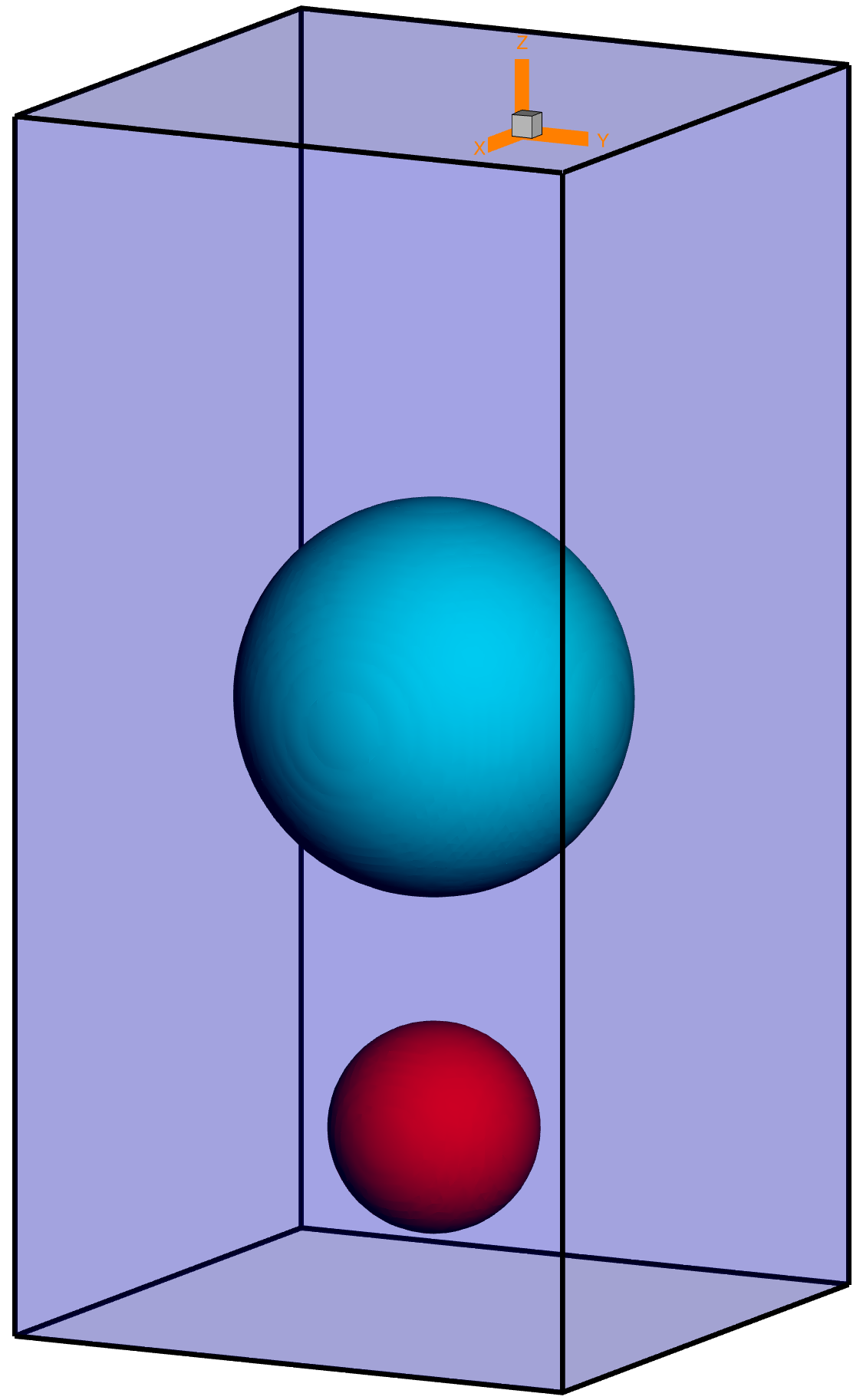}
	}\quad
	\subfigure[$t = 3$.]{	
		\includegraphics[width=0.18\textwidth, height=0.29\textwidth]{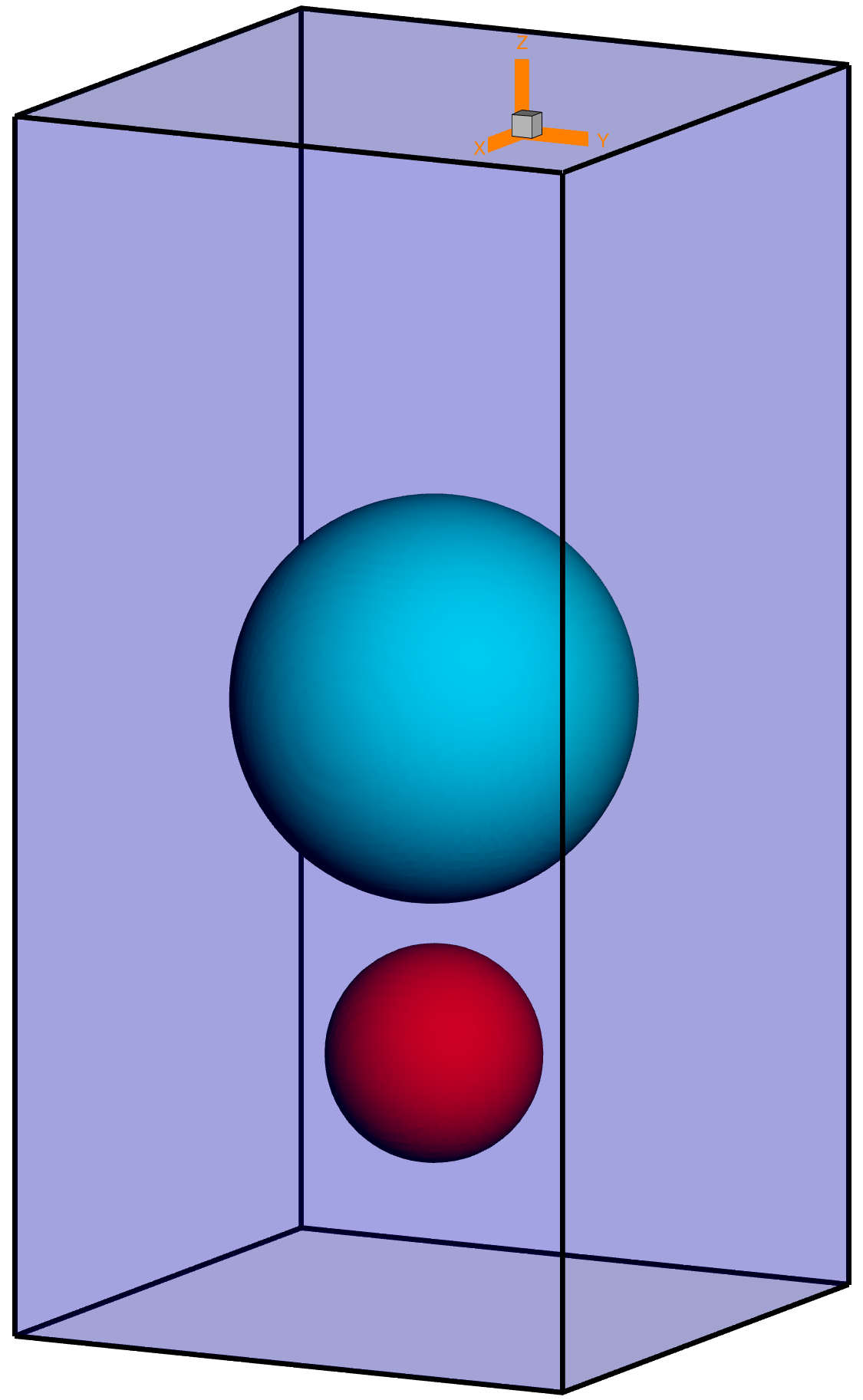}
	}\quad
	\subfigure[$t = 5$.]{	
		\includegraphics[width=0.18\textwidth, height=0.29\textwidth]{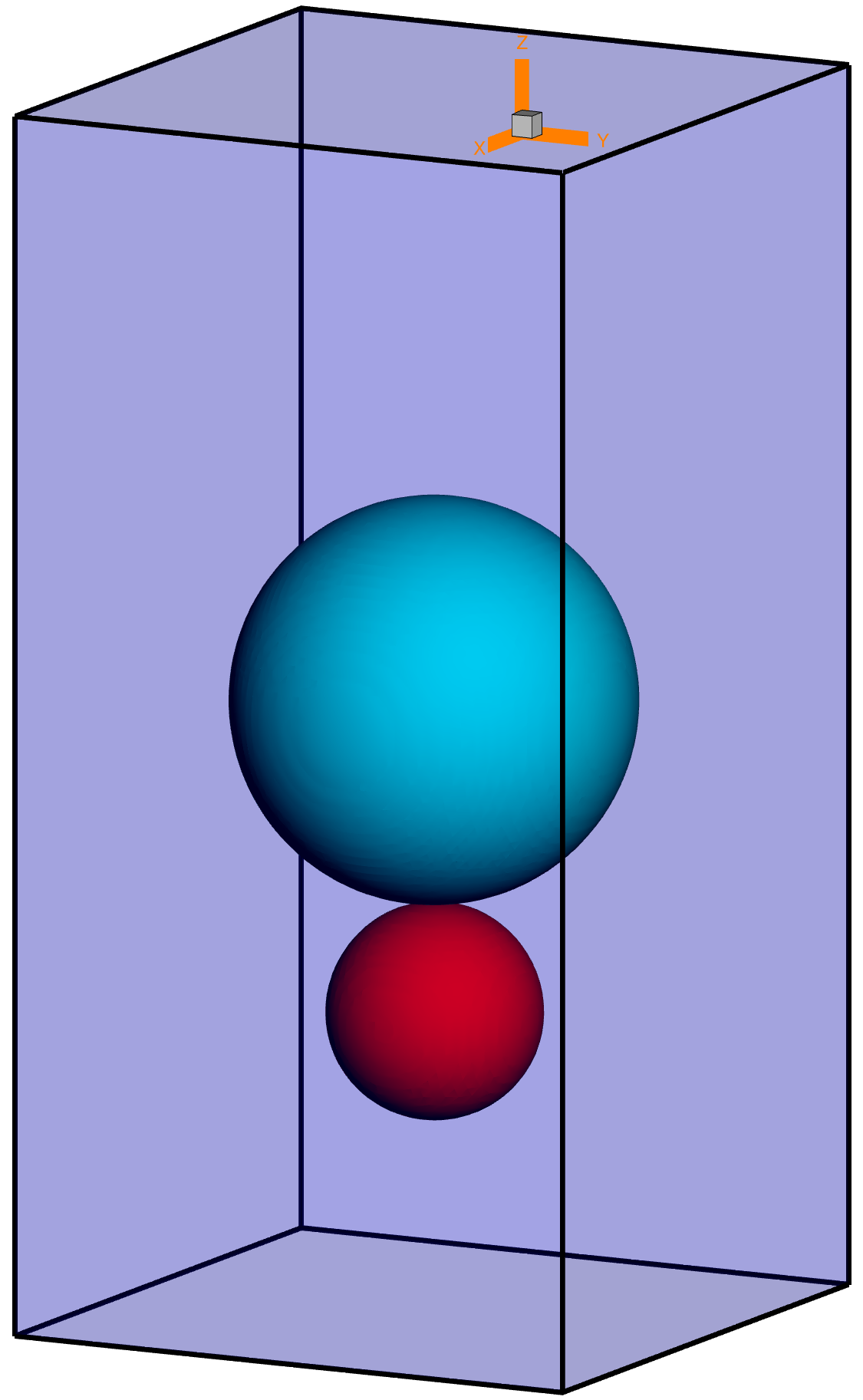}
	}\\
	\subfigure[$t = 8$.]{		
		\includegraphics[width=0.18\textwidth, height=0.29\textwidth]{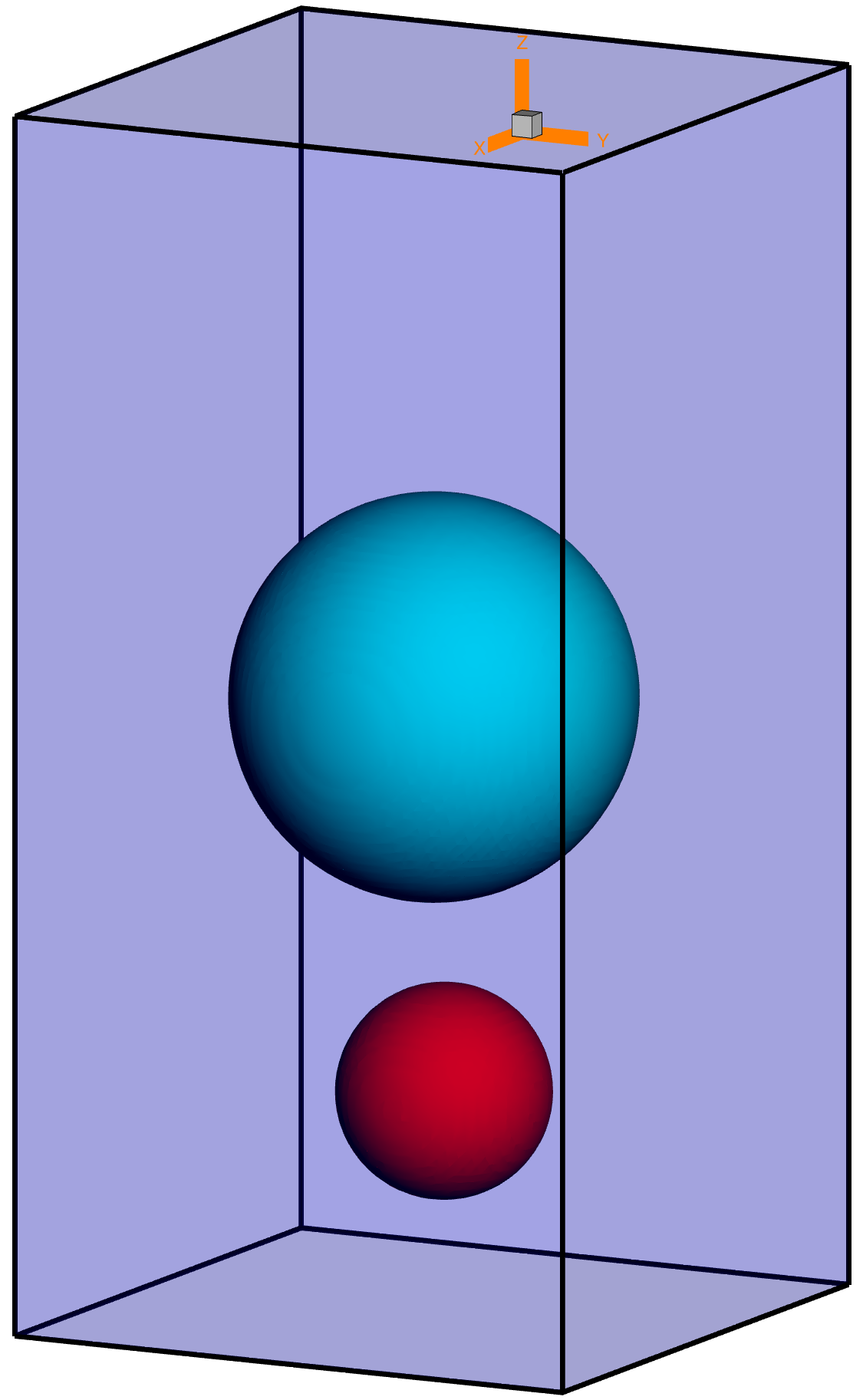}
	}\quad
	\subfigure[$t = 10$.]{	
		\includegraphics[width=0.18\textwidth, height=0.29\textwidth]{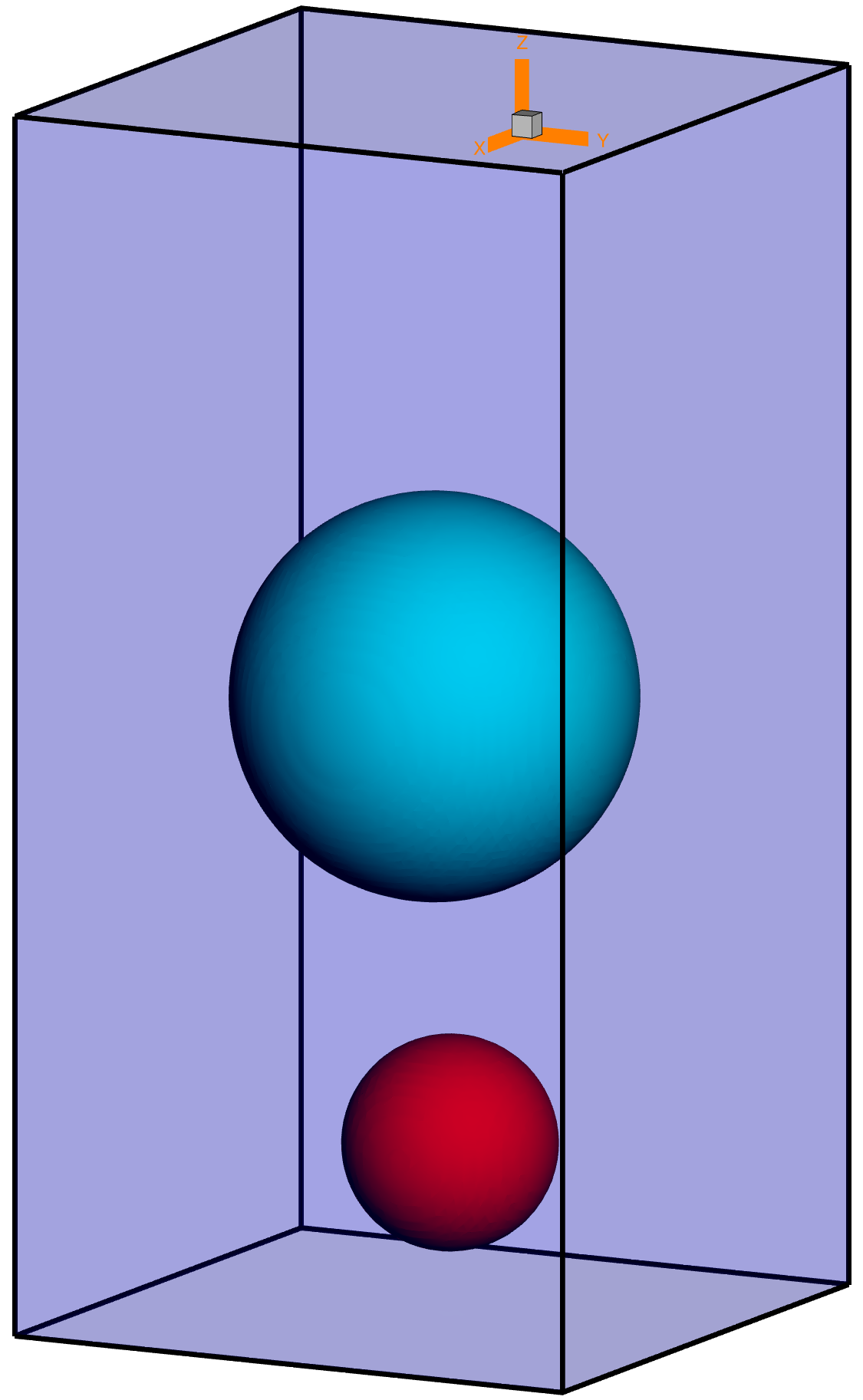}	
	}	
	\caption{Time evolution of a spherical active particle colliding with a fixed spherical obstacle by self-propulsion at different times. The isosurfaces of $\{\phi_1 = 0.5\}$ (active particle) and $\{\phi_2 = 0.5\}$ (obstacle) at $t = 0, 3, 5, 8, 10$ are shown in red, and the blue color represents the background fluid matrix.\label{fig:3DFSI}}
\end{figure}
This example demonstrates the power of the model and the numerical code in simulating complex fluid-structure interaction scenarios.

\section{Concluding remarks}

In this paper, we develop a new computational modeling framework for developing mathematical models to numerically study fluid-structure interaction with solid (rigid or elastic)  particles immersed in a viscous fluid matrix using the phase field embedding approach.
When the particle is rigid, its  state and domain is described by a zero velocity gradient tensor together with  a phase field that defines its profile. A hybrid thermodynamically consistent hydrodynamic model is then derived for the fluid-particle ensemble  by the generalized Onsager principle. When the particle is elastic, the zero velocity gradient tensor constraint is replaced by a constitutive equation valid within the particle. The hyperelastic model reduces to the one for the rigid body in the limit of an infinite elastic modulus. Subsequently, we devise two thermodynamically consistent (i.e. structure-preserving) coupled and decoupled  schemes to solving these models numerically. The newly  proposed schemes with an added  stabilization mechanism are second order accurate with an enhanced stability. Other high order structure-preserving schemes can be devised using Backward Difference (BDF) methods or Runge-Kutta methods.
Mesh refinement tests are carried out to confirm convergence rates of the new schemes.
Several 2D and 3D numerical simulations are conducted to illustrate the thermodynamically consistency and
usefulness of the newly developed model and accompanying schemes in studying flow-active particle interactions. Extension to include other particle's material properties ( viscoelastic, soft etc.) and more high order numerical strategies will be reported in sequels.

\section*{Acknowledgments}
Qi Hong' work is partially supported by the China Postdoctoral Science Foundation through Grant 2020M670116, the Foundation of Jiangsu Key Laboratory for Numerical Simulation of Large Scale Complex Systems (202001), National Natural Science Foundation of China (award 11971051 and  NSAF-U1930402).
Qi Wang' work is partially supported by National Science Foundation of US (award DMS-1815921 and  OIA-1655740) and a GEAR award from SC EPSCoR/IDeA Program.

\end{document}